\documentclass[a4paper,11pt]{article}
\usepackage{latexsym}
\usepackage{amsmath, amsthm, amssymb, bbm}
\usepackage[mathscr]{eucal}
\usepackage{fancyhdr}
\usepackage{tikz}
\usepackage{authblk}
\theoremstyle{plain}
\newtheorem{proposition}{Proposition}[section]
\newtheorem{theorem}[proposition]{Theorem}

\theoremstyle{definition}

\theoremstyle{remark}
\newtheorem{remark}[proposition]{Remark}

\newcommand{\rhs}{r.h.s.\ }
\newcommand{\lhs}{l.h.s.\ }
\newcommand{\wrt}{w.r.t.\ }
\newcommand{\cf}{cf.\ }

\newcommand{\ud}{\mathrm{d}}
\newcommand{\del}{\partial}
\newcommand{\eps}{\varepsilon}

\DeclareMathOperator{\supp}{supp}

\newcommand{\betrag}[1]{{\lvert #1 \rvert}}

\newcommand{\R}{\mathbb{R}}
\newcommand{\C}{\mathbb{C}}

\newcommand{\skal}[2]{\left< #1 , #2 \right>}
\newcommand{\order}{\mathcal{O}}
\newcommand{\id}{\mathrm{id}}
\newcommand{\1}{\mathbbm{1}}

\newcommand{\diag}{\mathrm{diag}}

\newcommand{\WDp}[1]{\colon \negthickspace #1 \! \colon \negthickspace }

\DeclareMathOperator{\WF}{WF}

\DeclareMathOperator{\Had}{Had}
\newcommand{\N}{\mathbb{N}}

\newcommand{\ket}[1]{\lvert #1 \rangle}
\newcommand{\braket}[2]{\langle #1 \vert #2 \rangle}

\newcommand{\Pei}[2]{\lfloor #1, #2 \rfloor}
\newcommand{\sst}[1]{\scriptscriptstyle{#1}}  % small font for the subscripts

\newcommand{\CatSub}{\mathbf{Sub}}
\newcommand{\CatVec}{\mathbf{Vec}}
\newcommand{\CatAlg}{\mathbf{Alg}}

\newcommand{\CatdgA}{\mathbf{dgA}}

\newcommand{\HS}{\mathcal{H}}

\newcommand{\F}{\mathfrak{F}}
\newcommand{\A}{\mathfrak{A}}
\newcommand{\E}{\mathfrak{E}}
\newcommand{\D}{\mathfrak{D}}
\newcommand{\g}{\mathfrak{g}}
\newcommand{\V}{\mathfrak{V}}
\newcommand{\BV}{\mathfrak{BV}}
\newcommand{\CE}{\mathfrak{CE}}
\newcommand{\X}{\mathcal{X}}
\newcommand{\Scal}{\mathcal{S}}
\newcommand{\Y}{\mathcal{Y}}

\newcommand{\Tens}{\mathfrak{Tens}}

\newcommand{\Poin}{\mathcal{P}}
\newcommand{\Cauchy}{\mathcal{C}}

\DeclareMathOperator{\ran}{ran}
\DeclareMathOperator{\Lap}{\bigtriangleup}

\newcommand{\pg}{\mathrm{pg}}
\newcommand{\af}{\mathrm{af}}

\newcommand{\gh}{\mathrm{gh}}

\newcommand{\mc}{{\mu\mathrm{c}}}
\newcommand{\ml}{\mathrm{ml}}
\newcommand{\reg}{\mathrm{reg}}
\newcommand{\ren}{\mathrm{r}}
\newcommand{\loc}{\mathrm{loc}}

\newcommand{\nm}{\mathrm{nm}}
\newcommand{\ia}{\mathrm{int}}
\newcommand{\ext}{\mathrm{ext}}

\newcommand{\al}{\alpha}
\newcommand{\Ga}{\Gamma}
\newcommand{\De}{\Delta}
\newcommand{\La}{\Lambda}
\newcommand{\ph}{\varphi}

\newcommand{\DDp}{\Gamma'_{\Delta_{D}}}
\newcommand{\DD}{\Gamma_{\Delta_{D}}}

\newcommand{\DCp}{\Gamma'_{\Delta}}

\newcommand{\T}{\cdot_{{}^\mathcal{T}}}

\newcommand{\TR}{\cdot_{{}^{\TTR}}}

\newcommand{\delTo}{\delta^{\sst{\TT}}_{S_0}}

\newcommand{\TT}{\mathcal{T}}
\newcommand{\TTR}{\mathcal{T}_\ren}
\newcommand{\TRH}{\cdot_{{}^{\TTR}}}
\newcommand{\expT}[1]{e^{#1}_{\sst{\TT}}}
\newcommand{\expTR}[1]{e^{#1}_{\sst{\TTR}}}

\begin{document}

\title{The effective theory of strings}
%\author{Dorothea Bahns\inst{1} \and Katarzyna Rejzner\inst{2} \and Jochen Zahn\inst{1}}
\author[1]{Dorothea Bahns\thanks{bahns@uni-math.gwdg.de}}
\author[2]{Katarzyna Rejzner\thanks{katarzyna.rejzner@desy.de}}
\author[1]{Jochen Zahn\thanks{jzahn@uni-math.gwdg.de}}

%\institute{Courant Research Centre ``Higher Order Structures'', University of G\"ottingen, \\ Bunsenstra{\ss}e 3--5, 37073 G\"ottingen, Germany. \email{bahns, jzahn @uni-math.gwdg.de} \and II.\ Institut f\"ur Theoretische Physik, Universit\"at Hamburg, Luruper Chaussee 149, \\ 22761 Hamburg, Germany. \email{katarzyna.rejzner@desy.de}}
\renewcommand\Affilfont{\itshape\small}
\affil[1]{Courant Research Centre ``Higher Order Structures'', University of G\"ottingen, Bunsenstra{\ss}e 3--5, 37073 G\"ottingen, Germany}
\affil[2]{II.\ Institut f\"ur Theoretische Physik, Universit\"at Hamburg, Luruper Chaussee 149, 22761 Hamburg, Germany}

\maketitle

\begin{abstract}
We show that the Nambu--Goto string, and its higher dimensional generalizations, can be quantized, in the sense of an effective theory, in any dimension of the target space. The crucial point is to consider expansions around classical string configurations. We are using tools from perturbative algebraic quantum field theory, quantum field theory on curved spacetimes, and the Batalin--Vilkovisky formalism. Our model has some similarities with the L\"uscher--Weisz string, but we allow for arbitrary classical background string configurations and keep the diffeomorphism invariance.
\end{abstract}

\section{Introduction}

It is part of the lore of string theory that bosonic strings in an ambient target space $M$ can be quantized consistently only when the dimension of $M$ is $26$. Essentially, the need for this critical dimension can be traced to the representation theory of the Virasoro algebra, whose generators are identified with the mode expansions of the constraints governing the system. In the framework of the Polyakov action, this constraint is the requirement that the energy-momentum tensor should vanish \cite{GreenSchwarzWittenI}, whereas in a covariant quantization of the Nambu--Goto action, it stems from the diffeomorphism invariance \cite{Rebbi74}. 

While the study of the proposed extra dimensions has triggered numerous interesting developments especially in mathematics, there have been indications in the past that there are quantization schemes for the Nambu--Goto string which do not require the ambient space's dimension to be fixed~\cite{Pohlmeyer,ThiemannLoopString,MeusburgerRehren}. In the present paper we propose yet another quantization scheme which does not require such a critical dimension. It is based on  perturbative methods in the framework of (generally) covariant field theories, and is also valid for higher dimensional objects (``branes'').

Historically, the starting point of string theory was the dual resonance model, which was then cast into the Nambu--Goto action. This action is the straightforward generalization of the action of the free relativistic (0-dimensional) particle moving in a target space $M$ to a 1-dimensional extended object.
While the former measures the geodesic length of the particle's worldline, the latter measures the world-sheet volume swept out by the string. In this sense, the Nambu--Goto action is geometric, and it is in particular invariant under changes of the worldsheet's parametrization. Just as one can pass from the length functional to the energy functional, one can pass from the Nambu--Goto action to the (quadratic) Polyakov action\footnote{In the case of branes one passes to the Polyakov-Howe action.}. It has the same classical solutions, but that does not automatically entail that the quantum theories are in some sense `equivalent', since in the quantum theory, also off-shell configurations have to be taken into account.
%\marginpar{Begin new}
In particular, the expectation value for the position of the string in the usual vacuum state is constant, i.e., it describes an event in the target space, and not a two-dimensional submanifold. For such configurations, the Nambu--Goto action is singular, which casts more doubts on the equivalence of the Nambu--Goto and the Polyakov action. In our approach, the expansion around submanifolds of the correct dimension is crucial for the absence of anomalies.
%\marginpar{End new}
Moreover, the Polyakov action introduces new degrees of freedom (a non-physical metric on the worldsheet), and new symmetries (conformal invariance).
As shown in \cite{BrandtTroostVanProeyen96,GomisParisSamuel}, it is the ghost corresponding to this new symmetry which is responsible for the anomalies of the Polyakov action (absent only in 26 dimensions). This symmetry is not present in our model, which is anomaly-free in any dimension of the target space.

On the other hand, also the canonical quantization of the Nambu--Goto action \cite{Rebbi74} requires the critical dimension $n = \dim M= 26$. 
However, this quantization scheme has some defects that are usually not mentioned in the literature. To see this, let us recall the setting.
The starting point is a parametrization $\tilde X$ of the worldsheet (immersed in the target Minkowski space $M$), whose components in Cartesian coordinates (at fixed parameter time $\tau = 0$) $\tilde X^a$, $0 \leq a \leq n-1$ are interpreted as the position variables. One then chooses the so-called orthonormal gauge and computes the corresponding momenta $P^a$. Expanding in Fourier modes, one obtains, for the open string,
\begin{align*}
 \tilde X^a(\sigma) & = q^a + \sum_{n = 1}^\infty \frac{1}{n} \cos n \sigma \left( \alpha_n^a + {\alpha_n^a}^* \right), \\
 P^a(\sigma) & = p^a + \sum_{n = 1}^\infty i \cos n \sigma \left( - \alpha_n^a + {\alpha_n^a}^* \right),
\end{align*}
where $\sigma$ is the worldsheet coordinate supplementing $\tau$, $q^a$ is the center of mass at $\tau = 0$, $p^a$ is the total momentum, and the $\alpha$'s are higher order Fourier mode coefficients. In the canonical quantization scheme, one then has the commutation relations
\begin{align}
\label{eq:Comm_qp}
 [q^a, p^b] & = i  h^{ab}, \\
 [\alpha_n^a, {\alpha_m^b}^*] & = n \delta_{mn} h^{ab}, \nonumber
\end{align}
where $h$ is the metric on the target Minkowski space. Hence, after a change of their normalization, the $\alpha^*$'s and $\alpha$'s are represented as creation and annihilation operators on a Fock space, and the $p^a$ and $q^a$ are represented in the Schr\"odinger representation on $L^2(\R^n)$.
Physical states have, among others, to fulfill the condition $(L_0 - a_0) \ket{\psi} = 0$, where
\[
 L_0 = \tfrac{1}{2} p^2 + \sum_{n=1}^\infty {\alpha_n^a}^* h_{ab} \alpha_n^b,
\]
and $a_0$ is a real number. It follows that physical states must be elements of $L^2(\R^n)$, supported on a discrete set of mass hyperboloids. But this set is empty, so there are no physical states.\footnote{For this conclusion it is crucial that we have canonical commutation relations \eqref{eq:Comm_qp}. This is the, usually neglected, difference to dual resonance models, where the $p^a$ are central.}
%\marginpar{Changed Begin}
This problem seems to have been first mentioned in \cite{GrundlingHurst}\footnote{We thank Peter Bouwknegt for pointing out this reference to us.}. It is also treated in \cite{DimockString2}, where the Hilbert space is modified to match the physicality condition. However, the problem of representing the $q$'s on this Hilbert space is left open. While there certainly are sensible representations, it is clear that these will violate \eqref{eq:Comm_qp}.
%\marginpar{Changed End}
%Surely there are ways to cure this problem, but to the best of our knowledge, this is not discussed in the literature.

This problem is absent in the so-called light cone gauge. There, one chooses a time-like direction in target space and fixes the parametrization completely by setting $\tilde X^+(\sigma) = 0$ and $P^+(\sigma) = p^+$. Hence, one can eliminate the canonical pair $\{q^+, p^-\}$ (and the oscillator modes ${\alpha^\pm_m}^{(*)}$), using the Dirac bracket instead of the Poisson bracket, \cf \cite[Sec.\ 12.5.4]{BrinkHenneaux}. But this singles out a preferred direction, as states can now be arbitrarily well localized in both $q^+$ and $p^-$, which is not possible for the other canonical pairs $\{ q^-, p^+ \}, \{ q^i, p^i \}$, for $i$ in the transverse direction. We note that this statement is not in contradiction with the fact that the Lorentz algebra can be represented in the case $n=26$, as this representation is non-linear, i.e., $\delta_\Lambda q^a \neq \Lambda^a_b q^b$.

Also our approach breaks Lorentz invariance, but this is caused by spontaneous symmetry breaking, i.e., the choice of a ``ground state'', and not by choosing a parametrization, which should be physically irrelevant. Furthermore, we can ensure the preservation of the Poincar\'e symmetry in the renormalized theory.

%\marginpar{New paragraph}
Our main motivation for the study of the Nambu--Goto action is that it exhibits diffeomorphism invariance, so it may be seen as a toy model for quantum gravity. Nevertheless, there might be physical applications, for example the description of vortex lines in QCD, or, in the case of membranes, of phase boundaries. Another possibility would be to consider a four-dimensional membrane embedded into a higher-dimensional space as a model for space-time. The Nambu--Goto action then plays the role of the cosmological constant.

Let us now sketch our construction. Our aim is to interpret string theory as a geometric theory of Lorentzian submanifolds of a Lorentzian target space $M$. We are working perturbatively, and the main conceptual step is to realize what this actually means in the context of quantization of submanifolds of a spacetime $M$. In perturbation theory, one chooses some classical solution, referred to as the background in the following, and considers small fluctuations around it. In the present context, a background is itself a submanifold $\Sigma$, and the actual submanifold $\tilde \Sigma$, referred to as the dynamical submanifold, can be parametrized by smooth sections $\varphi \in \Gamma^\infty(\Sigma, TM)$, where one uses the exponential map
of $M$ to map $\Sigma$ to $\tilde \Sigma$. Concretely, we consider
\[
\tilde X(x) = \exp_{X(x)}(\lambda \varphi(x)),
\]
where $X$ is the embedding $\Sigma \to M$, $x \in \Sigma$, and $\lambda$ is a length scale. Expanding the Nambu--Goto action in $\lambda$, one obtains a classical field theory on $\Sigma$. As by the expansion in $\lambda$ we obtain terms of arbitrarily high order, the quantized theory has to be understood in the sense of an effective one.

This classical field theory still has reparametrization invariance, as any diffeomorphism of $\Sigma$ induces a reparametrization of $\tilde \Sigma$, i.e., a change of the section $\varphi$ that leaves the image invariant. In order to deal with this symmetry, we employ the Batalin--Vilkovisky (BV) formalism. For the sake of simplicity in concrete calculations, we will at some point restrict ourselves to flat target manifolds, i.e., Minkowski space.

For perturbation theory to make sense, we have to restrict to background submanifolds $\Sigma$ that are on-shell\footnote{More precisely, the deviation from being on-shell must be at least first order in the expansion parameter.}, i.e., that locally fulfill the Euler--Lagrange equations corresponding to the Nambu--Goto Lagrangean. 
We do not consider the theory on some fixed background, but consider all backgrounds simultaneously, in a coherent way. This is done in the spirit of the general covariant locality principle \cite{BrunettiFredenhagenVerch}. Mathematically, one uses the language of category theory.

The quantization is performed in the setting of perturbative algebraic quantum field theory, i.e., via deformation quantization. By construction, one obtains the background configuration $X(x)$ as the expectation value of $\tilde X(x)$ in any quasi-free state of the free field algebra\footnote{Let us note that in the conventional approach discussed above, the expectation value of $\tilde X^a(\sigma)$ in a non-excited state (more generally a state with fixed excitation number for each mode) does not depend on $\sigma$. In this sense, the ground state does not correspond to a string, but to a particle. As the quantization is anomaly-free in our setting, it is tempting to speculate that the anomaly found in the standard approach is due to an inappropriate choice of the vacuum.} (at least in the case of a flat background, in the general case at first order in $\lambda$). A crucial prerequisite for the treatment of interacting theories is the existence of Hadamard functions compatible with the gauge symmetry. We show that such Hadamard functions always exist for the open string and explicitly construct a representation of the free field algebra of the Dirichlet string with a positive definite physical Hilbert space.

We renormalize the theory in the Epstein--Glaser framework, i.e., by extension of distributions. Here it is crucial to do the renormalization in a manner that is covariant \wrt coordinate changes on the background $\Sigma$. This is achieved by techniques from quantum field theory on curved spacetimes \cite{HollandsWaldTO}. These methods also ensure that we can quantize in a way that respects the Poincar\'e covariance of the target Minkowski space. Here, covariance means covariance under simultaneous transformations of the background $\Sigma$ and the section $\varphi$ describing $\tilde \Sigma$. This guarantees that the Poincar\'e symmetry is preserved in the renormalized theory.

The question as to whether the reparametrization invariance can be kept in the renormalized theory can be answered by studying the cohomology of the Batalin--Vilkovisky complex at ghost number one \cite{BarnichBrandtHenneaux00,RejznerFredenhagenQuantization}. We will show that this cohomology is trivial, so that no anomalies occur.
This follows almost trivially from the fact that $\ud X$ has maximal rank, which highlights the importance of expanding around a background configuration of the correct dimension.

The idea to consider fluctuations off a classical background string is also the basis of the approach of L\"uscher and Weisz \cite{LuscherWeisz}. However, they only allow for transversal fluctuations, i.e., for a flat background
\begin{equation}
\label{eq:flatBackground}
 \R \times [0,\pi]  \ni (\tau, \sigma) \mapsto X(\tau, \sigma) = (\tau, \sigma, 0, \dots, 0) \in M
\end{equation}
they parametrize deviations as
\[
 \tilde X(\tau, \sigma) = (\tau, \sigma, \phi^1(\tau, \sigma), \dots, \phi^{n-2}(\tau, \sigma)),
\]
with, e.g., Dirichlet boundary conditions $\phi^j(\tau, 0) = \phi^j(\tau, \pi) = 0$.
Obviously, one can not parametrize a generic Dirichlet string in this way, as foldings in the transversal directions are not possible. That the physical degrees of freedom are the transversal fluctuations is put in by hand in this approach, whereas in our setting it is ensured, at the linearized level, by the BV formalism. In this context, we point out that the restriction to transversal fluctuations corresponds to the abelianization of the gauge transformation of the free action, \cf \cite[Sec.~17.4.1]{HenneauxTeitelboim}. As noted there, the corresponding field redefinitions are in general nonlocal. This is not problematic in an effective field theory setting, but one expects that it is rather tedious to identify the admissible forms of the higher order terms. We refer to \cite{AharonyDodelson,DubovskyFlaugerGorbenko} for recent discussions of this topic. In our approach, the higher order terms are completely fixed by the Nambu--Goto action and, at each order, a finite number of renormalization constants multiplying terms with straightforward geometric interpretation. The exact correspondence between the L\"uscher--Weisz string and our model is certainly a very important question to investigate in the future.
%\marginpar{New sentence}
See also Remark~\ref{rem:Kleinert} below for some comments on this issue.

Whether there is any relation to the effective string theory of Polchinski and Strominger \cite{PolchinskiStrominger} is much less clear. There, one expands in powers of $R^{-1}$, where $R$ is the radius of the compactified dimension around which the string is wrapped. Clearly, such an expansion differs from our expansion around a classical string configuration, so at least at first sight, there is no obvious connection. Furthermore, the action  considered there is not obtained by the expansion of the Nambu--Goto action, but by adding terms to the Polyakov action. For a recent attempt to clarify the relation between the L\"uscher--Weisz and Polchinski--Strominger effective string theories, see \cite{DubovskyFlaugerGorbenko}.

%\marginpar{New paragraph}
In our approach, the observables are localized, i.e., compactly supported, and local. This is in contrast to the observables in Pohlmeyer's approach \cite{MeusburgerRehren}, which are neither localized nor local.

We would like to point out that the observables in our model are parametrized by test sections on the background\footnote{Even though these test sections may be obtained from test functions on the target space, \cf Section~\ref{sec:Fields}.} and are local \wrt the induced metric on the background. Hence, our model is not related to the string-localized fields occurring in infinite spin representations of the Poincar\'e group \cite{MundSchroerYngvason04}.

The main conceptual open problem is to show the (perturbative) background independence of the framework, in the sense defined in \cite{BrunettiFredenhagenQG} or \cite{HollandsWaldStress}.

The article is organized as follows: In the next section, we introduce the setup, i.e., the categorial description, the relevant test function and configuration spaces, and the observable algebras. In Section~\ref{sec:Diffeo}, we discuss the diffeomorphism symmetry and the resulting BV complex. In Section~\ref{sec:GaugeFixing}, we perform the gauge fixing and discuss the resulting structure of the Peierls bracket and the BRST current. Physical observables (fields) are defined in Section~\ref{sec:Fields}. The quantization is performed in Section~\ref{sec:Quantization}. The relevant aspects of the cohomology of the BV complex are discussed in Section~\ref{sec:Cohomology}.

\section{The setup}
\label{sec:Setup}

We consider $d$-dimensional ($d \geq 2$) submanifolds (worldsheets) of a fixed time-oriented globally hyperbolic spacetime $(M, h)$ (the target space) with signature $(-, +, \dots, +)$ and dimension $n$.
We also require $M$ to be geodesically convex and complete\footnote{This assumption can be dropped if one works in a fully perturbative setting throughout, \cf Remark~\ref{rem:GeodesicCompleteness}.}, i.e., for each $p \in M$ the exponential map $\exp_p: T_p M \to M$ is a bijection.  
The submanifolds are required to have non-degenerate metric with hyperbolic signature. Furthermore, they inherit the time-orientation of $M$ and are required to be orientable. A particular instance are strings ($d=2$). In the following we use greek letters for the indices corresponding to vectors in $T \Sigma$, and roman ones for vectors in $T M$.
%As, for convenience, we want to switch between the perturbative and the nonperturbative setting, in particular in Section~\ref{sec:Cohomology}, we stick to this requirement.}

Given coordinates $x$ on $\Sigma$, the Nambu--Goto Lagrangean density is given by the induced volume form on $\Sigma$, i.e.,
\[
 \ell_\mathrm{NG} = \sqrt{-g} \ud x,
\]
where $g$ is the determinant of the induced metric on $\Sigma$, $g = X^* h$, or, written in local coordinates,
\[
 g_{\mu \nu}(x) = \ud X^a_\mu(x) h_{ab}(X(x)) \ud X^b_\nu(x).
\]
Here $X: \Sigma \to M$ is the embedding and $\ud X$ the corresponding differential. We say that a submanifold $\Sigma$ is \emph{on-shell} if it locally\footnote{We do not impose any boundary conditions.} fulfills the corresponding Euler--Lagrange equations
\[
 \del_\mu \left( \sqrt{-g} g^{\mu \nu} \ud X^a_\nu \right) = 0.
\]

When perturbatively expanding around such a $\Sigma$, we call it the \emph{background}. However, we do not fix a single one, but consider all possible backgrounds simultaneously,
in the spirit of \cite{BrunettiFredenhagenVerch}. This is encoded in the language of category theory. To this end, we define the category $\CatSub$ by\footnote{Note that we deviate here from the usual setup of the general covariant locality principle \cite{BrunettiFredenhagenVerch}. There, the morphisms are isometric, causal embeddings. Here, $\Sigma \to \Sigma'$ only if $\Sigma$ is a submanifold of $\Sigma'$, because we consider distinct submanifolds $\Sigma$ and $\Sigma'$ as physically distinct, even if they are isometric (the position in the target space matters). Covariance under embeddings corresponding to global symmetries of $M$ is discussed in Section~\ref{sec:Poincare}.}
\begin{description}
\item[$\CatSub$:] The objects are $d$-dimensional, globally hyperbolic, oriented, on-shell submanifolds of $M$. If $\Sigma \subset \Sigma'$, with compatible orientation, and $\gamma \subset \Sigma$ for all causal paths $\gamma$ in $\Sigma'$ connecting two arbitrary elements $x,y \in \Sigma$, then there is a morphism $\chi: \Sigma \to \Sigma'$.
\end{description}
The restriction to globally hyperbolic submanifolds (see \cite{WaldGR} for a definition) is done to ensure the well-posedness of the Cauchy problem.

\begin{remark}
It is well-known that solutions of the minimal volume problem in hyperbolic signature in general develop singularities, c.f., for example, \cite{EggersHoppe} and references therein.\footnote{As shown in \cite{MuellerMinimalSurfaces}, the Cauchy problem for the closed Nambu--Goto string is well-posed in the following sense: Solutions exist locally on the source space, but globally in target space (they intersect all Cauchy surfaces). However, these solutions are in general not  immersions, only locally. For membranes which are initially close to being flat, there are results ensuring global existence \cite{AllenAnderssonIsenberg}.}
An example is the well-known oscillating cylinder solution
\[
 (\tau, \sigma) \mapsto (\tau, \cos \tau \cos \sigma, \cos \tau \sin \sigma, 0, \dots, 0),
\]
which has self-intersections at $\tau = (k+1/2) \pi$ at which the induced metric degenerates. For our approach, the requirement of a globally hyperbolic induced metric is crucial. Hence, we will in general have to restrict the classical solution. In the above example, one could simply restrict to $\tau \in (-\pi, \pi)$. In this sense, our approach treats the local excitations of minimal surfaces or membranes. Also note that in our approach we can not describe a change of the topology, so we only consider the self-interaction of membranes.
\end{remark}
%\marginpar{Add remark on singularities here?}

Other submanifolds $\tilde \Sigma$ are now parametrized by $\varphi \in \Gamma^\infty(\Sigma, TM)$, i.e., smooth sections of the tangent bundle. We define 
\begin{align*}
 \tilde X : \Sigma & \to M \\
 x & \mapsto \exp_{X(x)}(\lambda \varphi(x)),
\end{align*}
where $\exp$ is the exponential map in $M$, and $\lambda$ is a formal expansion parameter of dimension length$^{d/2}$.
In principle, one should ensure that the image of $\ud \tilde X$ in $T_{\tilde X} M$ is $d$-dimensional and contains time-like directions. However, we will mostly work perturbatively, i.e., in the sense of formal power series in $\lambda$, and at zeroth order the condition is fulfilled by construction, as $\tilde X = X + \order(\lambda)$. Hence, no further restrictions on $\varphi$ are necessary. Even though it is not a submanifold in this perturbative setup, we will for convenience call the image $\tilde \Sigma$ of $\tilde X$ the \emph{dynamical submanifold}.

Other categories that we will need in the following are:
\begin{description}
\item[$\CatVec$:] The objects are locally convex vector spaces. The morphisms are continuous linear maps.
\item[$\CatVec_i$:] As $\CatVec$, but the morphisms are required to be injective maps.
\item[$\CatAlg^{(*)}$:] The objects are topological ($*$) algebras. The morphisms are continuous injective ($*$) algebra homomorphisms.
\item[$\CatdgA$:] The objects are differential graded algebras. The morphisms are injective differential graded algebra homomorphisms.
\end{description}
Obviously, by applying the forgetful functor, objects of $\CatAlg^{(*)}$ and $\CatdgA$ can also be considered as objects of $\CatVec_i$.

We define the configuration space as $\E(\Sigma) \doteq \Gamma^\infty(\Sigma, TM)$, equipped with the usual locally convex topology of uniform convergence of all derivatives on compact sets. This defines a contravariant functor $\E: \CatSub \to \CatVec$.
The morphism of $\CatVec$ corresponding to the morphism $\chi: \Sigma \to \Sigma'$ is the restriction $\chi^*(\varphi) = \varphi|_{\Sigma}$.

A functor between $\CatSub$ and $\CatVec_i$ is the functor $\E_c$ which associates to each submanifold $\Sigma$ the vector space $\E_c(\Sigma) \doteq \Gamma^\infty_c(\Sigma, T M)$ of smooth compactly supported test vector fields with the usual locally convex topology. This functor is covariant with
\[
 \chi_*(\varphi)(x') = \begin{cases} \varphi(x') & x' \in \Sigma, \\ 0 & \text{otherwise.} \end{cases}
\]
Another covariant functor between these categories is the functor $\mathfrak{D}$ which associates to each submanifold $\Sigma$ the vector space $\D(\Sigma) \doteq C^\infty_c(\Sigma)$ of smooth compactly supported test functions with the usual locally convex topology. Again, the image of a morphism $\chi: \Sigma \to \Sigma'$ is given by the canonical embedding of $C_c^\infty(\Sigma)$ in $C_c^\infty(\Sigma')$.

It remains to define the algebra $\F(\Sigma)$ of functionals on $\E(\Sigma)$. As in \cite{RejznerFredenhagen}, one first defines $\F_{\mathrm{loc}}(\Sigma)$ as the space of compactly supported smooth local functionals. The smoothness is understood in the sense of calculus on locally convex topological vector spaces, i.e., the derivative
of a functional $F$ at $\ph$ in the direction of $\psi$ is defined as
\begin{equation}
\label{de}
F^{(1)}(\ph)(\psi) \doteq \lim_{t\rightarrow 0}\frac{1}{t}\left(F(\ph + t\psi) - F(\ph)\right).
\end{equation}
$F$ is called differentiable at $\ph$ if $F^{(1)}(\ph)(\psi)$ exists for all configurations $\psi$. It is called continuously differentiable if it is differentiable at all points of $\E(\Sigma)$ and
$F^{(1)}:\E(\Sigma)\times \E(\Sigma)\rightarrow \C, (\ph,\psi)\mapsto F^{(1)}(\ph)(\psi)$
is a continuous map. It is called a $\mathcal{C}^1$-map if it is continuous and continuously differentiable. Higher derivatives and continuous differentiability are defined analogously. $F$ is smooth if it is $k$ times continuously differentiable for all $k$.

The support of a functional $F$ is defined as
\begin{multline*}
\supp F=\{ x \in \Sigma| \text{ for all open }U\ni x\ \exists \varphi,\psi\in\E(\Sigma), \supp\,\psi\subset U 
\\ \text{ such that }F(\varphi+\psi)\not= F(\varphi) \}.
\end{multline*}
A functional $F$ is \emph{local} if it can be written as
\[
 F(\varphi) = \int_\Sigma f(x, \varphi(x), \nabla \varphi(x), \dots ) \mu(x),
\]
where $f$ is a smooth function on the jet space, compactly supported with respect to the $x$ variable, and $\mu = \sqrt{-g} \ud x$ is the induced volume form on $\Sigma$. The space of all local functionals is denoted by $\F_{\mathrm{loc}}(\Sigma)$. Observe that 
$\F_{\mathrm{loc}}$ is a covariant functor from $\CatSub$ to $\CatVec_i$. The algebra $\F(\Sigma)$ is now defined as the algebraic completion of $\F_{\mathrm{loc}}(\Sigma)$ under pointwise multiplication, i.e.,
\[
 (F_1 F_2)(\varphi) = F_1(\varphi) F_2(\varphi).
\]
Elements of $\F(\Sigma)$ are called \emph{multilocal} functionals.
$\F$ is also a covariant functor to $\CatVec_i$, and even to $\CatAlg$.
From the definition of smoothness follows that for $F \in \F_{\mathrm{loc}}(\Sigma)$, and thus also $F \in \F(\Sigma)$, one has
\begin{align*}
 F^{(1)}(\varphi) & \in \E_c(\Sigma), & F^{(k)}(\ph) & \in \Gamma^\infty(\Sigma^k, TM^{\otimes k})',
\end{align*}
for all $\ph \in \E(\Sigma)$.
Since we can reconstruct smooth functionals from the Taylor series expansion and we work in a perturbative setting, it is convenient to write the multilocal functionals as
\[
F(\ph) = F(0) + \sum_{k=1}^\infty \frac{1}{k!} \skal{F^{(k)}(0)}{\ph^{\otimes k}},
\]
where $\skal{\cdot}{\cdot}$ denotes the evaluation of a compactly supported distribution on a smooth section. This can be rewritten, using a  formal notation based on the use of distributional kernels, as
\[
 \skal{F^{(k)}(0)}{\ph^{\otimes k}} = \int f(x_1,\ldots,x_k) \ph(x_1) \ldots \ph(x_k) \mu(x_1) \ldots \mu(x_k),
\]
where $f\in\Gamma^\infty(\Sigma^k, TM^{\otimes k})'$. Denoting by $\Phi_x$ the evaluation functional $\Phi_x: \ph \mapsto \ph(x)$, the above expression can be compactly written as
\begin{equation}\label{polynom}
F^{(k)}(0)=\int f(x_1,\ldots,x_k) \Phi_{x_1}\ldots\Phi_{x_k} \mu(x_1) \ldots \mu(x_k).
\end{equation}
Functionals of this form are usually called \emph{polynomials} and we can think of them as a ``basis'' in $ \F(\Sigma)$. Since the above notation is frequently used in physics, we will use it throughout this paper, keeping in mind the precise mathematical meaning of such expressions given above. The condition of locality for a polynomial (\ref{polynom}) can be rephrased as the requirement that $f(x_1,\ldots,x_k)$ is a distribution supported on the thin diagonal $\Delta^k(\Sigma)\doteq\left\{(x,\ldots,x)\in \Sigma^k:x\in \Sigma\right\}$ and its wave front set\footnote{We refer to \cite{BrunettiFredenhagenScalingDegree} for a definition of the wave front set and its use in perturbative algebraic quantum field theory.} is orthogonal to the tangent bundle of  $\Delta^k$, i.e., $\WF(f)\perp T\Delta^k(\Sigma)$.

On $\F(\Sigma)$ (more precisely on a suitable extension allowing for a gauge fixing, see below) we will define 
the classical Poisson structure via the so-called Peierls bracket, and then deform it to obtain the quantized algebra. However, it turns out that $\F(\Sigma)$ is not closed under the Peierls bracket, so a more general space of functionals is needed. As in \cite{BDF09}, we consider the space of \emph{microcausal} functionals $\F_{\mu\mathrm{c}}(\Sigma)$, i.e., the space of smooth, compactly supported functionals such that, for $k \geq 1$,
\begin{multline*}
 \WF(F^{(k)}(\varphi)) \cap \\  \{ (x_1, \dots, x_k; p_1, \dots p_k) \in T^* \Sigma^k \setminus \{ 0 \} | p_i \in \bar V_{x_i}^+ \ \forall i \text{ or } p_i \in \bar V_{x_i}^- \ \forall i \} = \emptyset.
\end{multline*}
Here $\bar V_x^\pm$ denotes the closure of the forward (backward) light cone in $T^*_x \Sigma$.
The space $\F_\mc(\Sigma)$ is important also in the context of quantization, since it already contains the Wick products. See \cite{BrunettiFredenhagenScalingDegree} for more details.

In order to be compatible with the condition on wavefront sets, one introduces on $\F_{\mu\mathrm{c}}(\Sigma)$ a topology that allows to control this property. On each space $\Gamma^\infty(\Sigma^k, TM^{\otimes k})'$  one introduces the so-called H\"ormander topology and $\F_{\mu\mathrm{c}}(\Sigma)$ is equipped with the initial topology with respect to all maps $F\mapsto F(\varphi)\in \C$ and $F\mapsto F^{(k)}(\varphi)\in\Gamma^\infty(\Sigma^k, TM^{\otimes k})'$, for all $k \geq 1$, $\varphi\in\E(\Sigma)$. We denote this topology by $\tau_\Xi$.
 Since  $\F_{\textrm{loc}}(\Sigma)\subset\F(\Sigma)\subset\F_{\mu\mathrm{c}}(\Sigma)$, we can equip the first two of these spaces with the topology
induced from $\F_{\mu\mathrm{c}}(\Sigma)$. It was shown in \cite{BDF09} that local functionals are dense in the space of microcausal functionals with respect to $\tau_\Xi$.

In the next step, one introduces the action. For this, we need the induced metric on the dynamical submanifold $\tilde \Sigma$,
\begin{equation*}
 \tilde g_{\mu \nu}(x) = \ud \tilde X^a_\mu(x) h_{ab}(\tilde X(x)) \ud \tilde X^b_\nu(x).
\end{equation*}
This is to be understood in the sense of formal power series in $\lambda$.
We note that $\tilde g^{\mu \nu}$ is the inverse of $\tilde g_{\mu \nu}$ and is not obtained by raising the indices with $g^{\mu \nu}$. Furthermore, $\tilde g$ denotes the determinant of $\tilde g_{\mu \nu}$.
Following \cite{RejznerFredenhagen}, the Lagrangean can be defined as a natural transformation between the functors $\D$ and $\F_\mathrm{loc}$, i.e., we define the Nambu--Goto Lagrangean as 
\begin{equation}
\label{eq:Nambu-Goto}
 L_\Sigma(f)(\varphi) = \lambda^{-2} \int_\Sigma f \frac{\sqrt{- \tilde g}}{\sqrt{-g}} \mu,
\end{equation}
for $f \in \D(\Sigma)$.
The action $S$ is an equivalence class of Lagrangeans, where
\begin{equation}\label{equivalence}
L_1\sim L_2\ \textrm{if}\ \supp (L_{1,\Sigma}-L_{2,\Sigma})(f)\subset\supp\, \ud f \ \forall \Sigma, f \in \D(\Sigma).
\end{equation}

We now characterize the symmetries of the action. To this end, we need one more structure, namely
the vector fields $\V(\Sigma)$ on $\E(\Sigma)$. In the present setting, $\V(\Sigma)$ can be interpreted as the space of smooth maps $\X: \E(\Sigma) \to \E_c(\Sigma)$, i.e., we define the tangent space $T_\ph \E(\Sigma)$ to be $\E_c(\Sigma)$ for all $\ph \in \E(\Sigma)$ (\cf \cite{KrieglMichor} for the definition of vector fields on infinite-dimensional manifolds). There is a natural action as derivatives on $\F(M)$ by
\begin{equation}\label{DerivationVF}
 (\del_\X F)(\varphi) = \skal{F^{(1)}(\varphi)}{\X(\varphi)}.
\end{equation}
For a vector field $\X$ depending polynomially on $\ph$, we can write (\ref{DerivationVF}) formally as
\[
\del_\X F=\int f_\X(x_1,\ldots,x_k,y)\Phi_{x_1}...\Phi_{x_k}\frac{\delta F}{\delta \ph(y)} \mu(x_1) \ldots \mu(y),
\]
with some distribution $f_\X \in \Gamma^\infty(\Sigma^{k+1}, TM^{\otimes(k+1)})'$.
In this sense, we can think of the functional derivatives $\frac{\delta}{\delta \ph(y)}$ as ``generators'' of  $\V(\Sigma)$. In the physics literature they are called \textit{antifields} and are denoted by $\Phi^\ddagger_y\equiv\frac{\delta}{\delta \ph(y)}$. The above polynomial vector field can now formally be written as
\[
\X=\int f_\X(x_1,\ldots,x_k,y)\Phi_{x_1}...\Phi_{x_k}\Phi^\ddagger_y \mu(x_1) \ldots \mu(y).
\]
The support of an $\X$ of this form is defined as\footnote{There is also a definition of the support that does not depend on the formal notation, \cf \cite{RejznerFredenhagen}.}
\[
 x \in \supp \X \Leftrightarrow \exists i \text{ s.t. } x \in \pi_i \supp f_\X,
\]
where $\pi_i$ are the canonical projections $\Sigma^{k+1} \to \Sigma$.
One then defines $\V(\Sigma)$ as the space of vector fields with compact support. It is an $\F(\Sigma)$ module and a covariant functor between the categories $\CatSub$ and $\CatVec_i$ by the pushforward
\[
 \chi_*(\X)(\varphi') = \chi_* \X(\chi^* \varphi').
\]
Here the outer pushforward on the \rhs is the pushforward of $\E_c$. 
The definition of locality and microcausality of such objects can be given as a corresponding condition on the support and wavefront set of $f_\X$. We refer to \cite{RejznerFredenhagen} for details.

As in the finite dimensional case, besides vector fields one can also consider multivector fields, which are graded symmetric tensor powers of $\V(\Sigma)$. We denote the space of such objects by $\Lambda \V(\Sigma)$, where
\[
 \Lambda \V(\Sigma) = \F(\Sigma) \oplus \V(\Sigma) \oplus \mathcal{C}_{\ml}^\infty(\E(\Sigma),\Lambda^2\E_c(\Sigma)) \oplus \dots .
\]
Here $\mathcal{C}_{\ml}^\infty$ denotes the space of smooth, compactly supported, multilocal maps. The polynomials in $\Lambda^m \V(\Sigma) \doteq \mathcal{C}_{\ml}^\infty(\E(\Sigma),\Lambda^m\E_c(\Sigma))$ can be written as
\begin{equation}\label{MultiVectorFields}
\X=\int f_\X(x_1,\ldots,x_k,y_1,\ldots,y_m)\Phi_{x_1} \ldots \Phi_{x_k}\Phi^\ddagger_{y_1}\wedge \ldots \wedge \Phi^\ddagger_{y_m}  \mu(x_1) \ldots \mu(y_m).
\end{equation}
On the space of multivector fields one can introduce the Schouten bracket (usually called \emph{antibracket} in physics). This is an odd graded Poisson bracket
\[
\{\cdot,\cdot\}:\La^n\V(\Sigma)\times\La^m\V(\Sigma)\to\La^{n+m-1}\V(\Sigma),
\]
which is graded antisymmetric
and satisfies the graded Leibniz rule.
For $\X\in\La^1\V(\Sigma)$ and $F\in\La^0\V(\Sigma)$ it coincides with the action of $\X$ as a derivation,
\[
\{\X,F\}=\partial_{\X}F\,,
\]
and for $\X,\Y\in\La^1\V(M)$ it coincides with the Lie bracket
\[
\partial_{\{\X,\Y\}}=\partial_\X\partial_\Y - \partial_\Y \partial_\X.
\]
Moreover, it satisfies the graded Jacobi rule.
Formally, using the notation of (\ref{MultiVectorFields}) we can write the antibracket as
\begin{equation}\label{antibracket}
\{\X,\Y\}=-\int\!\left(\!\frac{\delta \X}{\delta\Phi(x)}\wedge\frac{\delta \Y}{\delta\Phi^\ddagger(x)}+(-1)^{|\X|}\frac{\delta \X}{\delta\Phi^\ddagger(x)}\wedge\frac{\delta \Y}{\delta\Phi(x)}\!\right) \mu(x),
\end{equation}
where $|\X|$ is the grade of $\X$. This is the notation commonly used in physics, so we will also invoke it in the present paper. Nevertheless, one should keep in mind that the antibracket is a well defined geometrical object and all the expressions can be given a precise mathematical meaning, without referring to the formal notation.

\section{The diffeomorphism symmetry}
\label{sec:Diffeo}
From the point of view of physics, only the submanifold is important, not its parametrization. Now for any smooth bijection $f: \Sigma \to \Sigma$ we get another parametrization $\varphi'$ of $\tilde \Sigma$ by $\tilde X'(x) = \tilde X(f(x))$, i.e.,
\begin{equation}
\label{eq:varphi'}
 \varphi'(x) = \lambda^{-1} \exp_{X(x)}^{-1}(\tilde X(f(x))).
\end{equation}
In particular, such a transformation $\ph \mapsto \ph'$ is a symmetry of the action. 
We consider infinitesimal transformations and assume that the corresponding vector field on $\Sigma$ is compactly supported. 
Hence, the Lie algebra of gauge symmetries is given by $\g_c(\Sigma) \doteq \Gamma^\infty_c(\Sigma, T \Sigma)$. Again, this defines a covariant functor from $\CatSub$ to $\CatVec_i$.

We want to determine the infinitesimal change in $\varphi$ corresponding to this symmetry. Thus, we first need to know the action of a vector $\xi$ on the parametrization $\tilde X(x)$.

\begin{proposition}
\label{prop:dX}
For any $\xi \in T_x \Sigma$ we have
\[
 \ud \tilde X(x)(\xi) = \eta(1),
\]
where $\eta$ is the unique Jacobi field along the geodesic
\begin{equation}
\label{eq:alpha}
 [0,1] \ni s \to \alpha(s) = \exp_{X(x)}(s \lambda \varphi(x))
\end{equation}
with initial conditions $\eta(0) = \xi$, $\nabla_s \eta(0) = \lambda \nabla_\xi \varphi(x)$. The covariant derivative on the \rhs of the second equation is the covariant derivative in $M$.
\end{proposition}
The proof is given in the appendix.

We now need to find out which infinitesimal change in $\varphi(x)$ leads to the above infinitesimal change in $\tilde X(x)$. Thus, we want to determine $\delta_\xi \varphi(x) \in T_x M \simeq T_{X(x)} M$ such that
\[
 \ud \exp_{X(x)}(\lambda \varphi(x)) (\lambda \delta_\xi \varphi(x)) = \ud \tilde X(x)(\lambda \xi).
\]
With the usual identification $T_\eta T_{X(x)} M \simeq T_{X(x)} M$, the differential of the exponential on the \lhs is a map $T_{X(x)} M \to T_{\tilde X(x)} M$.
By \cite[Thm.~IX.3.1]{LangDiffGeo}, we have $\delta_\xi \varphi(x) = \nabla_s \eta(0)$, where $\eta$ is the unique Jacobi field along the geodesic defined in \eqref{eq:alpha}, satisfying the conditions
\begin{align*}
 \eta(0) & = 0, & \eta(1) & = \ud \tilde X(x)(\xi).
\end{align*}
We thus set
\[
 (\rho_{\Sigma}(\xi) \varphi)(x) \doteq \delta_\xi \varphi(x) = \nabla_s \eta(0).
\]
\begin{remark}
\label{rem:GeodesicCompleteness}
In order for the Jacobi field $\eta$ to exist and to be unique, one has to rely on the geodesic convexity and completeness of $M$. In the perturbative setup, one can drop this assumption, as by choosing $\lambda$ or $\ph$ small enough, one can always achieve the well-definedness of $\eta$. Expanding the thus obtained $\delta_\xi \varphi(x)$ in $\lambda$, one obtains a series depending only on the local geometric data at $x$ and $X(x)$. This series then makes sense, as a formal power series, for all configurations $\ph$.
\end{remark}

The above variation is a vector field, $\rho_{\Sigma}(\xi) \in \V(\Sigma)$, so there is a natural action $\del_{\rho_{\Sigma}(\xi)}$ on $\F(\Sigma)$. A natural structure associated with an action of a Lie algebra on a vector space is the Chevalley--Eilenberg complex. It characterizes the space of invariants by its 0-order cohomology.
In our case it is defined as
\[
 \CE(\Sigma) \doteq \mathcal{C}^\infty_\ml(\E(\Sigma),\La\g'(\Sigma)),
\]
where $\g(\Sigma) \doteq \Gamma^\infty(\Sigma, T\Sigma)$,  $\g'(\Sigma)$ is its dual, and $\Lambda$ again denotes the antisymmetric tensor powers. 
As in the previous section, we may write the polynomial elements of $\CE(\Sigma)$ in the form
\begin{equation}
\label{eq:F_CE}
F = \int f_F(x_1, \dots, x_l,y_1,\ldots,y_k) \Phi_{x_1} \dots \Phi_{x_l}C_{y_1}\wedge \dots \wedge C_{y_k}  \mu(x_1) \dots  \mu(y_k),
\end{equation}
where the $C_{x_i}$'s  are the evaluation functionals on $\g(\Sigma)$, i.e., $C_{x_i}^\mu(c)=c^\mu(x_i)$, for $c\in\g(\Sigma)$. These are called \emph{ghosts} in the physics literature. Here we keep the convention that field configurations are denoted by small letters ($\ph$, $c$), and evaluation functionals by capital ones. The condition of locality and compact support can now easily be formulated for the distribution $f_F$, exactly as for the polynomials of the form (\ref{polynom}).

Let $F \in  \CE^k(\Sigma)\doteq\mathcal{C}^\infty_\ml(\E(\Sigma),\La^k\g'(\Sigma))$, $\ph\in\E(\Sigma)$, and $\xi_0, \dots, \xi_k\in\g(\Sigma)$. On $\CE(\Sigma)$, one introduces a differential $\gamma_\Sigma: \CE^k(\Sigma)\rightarrow \CE^{k+1}(\Sigma)$ by
\begin{align*}
 (\gamma_{\Sigma} F)(\ph;\xi_0, \dots, \xi_k) & = \sum_{i = 0}^k (-1)^i (\del_{\rho_{\Sigma}(\xi_i)} F( \cdot ; \xi_0, \dots, \hat{\xi_i}, \dots, \xi_k))(\ph) \\
 & + \sum_{i < j} (-1)^{i+j} F(\ph;[\xi_i, \xi_j], \dots, \hat{\xi_i}, \dots, \hat{\xi_j}, \dots \xi_k),
\end{align*}
where the hat denotes omission.
This means that on elements of the form \eqref{eq:F_CE}, $\gamma_\Sigma$ acts as a graded left differential, whose action on the evaluation functionals is given by
\begin{align*}
 \gamma_\Sigma \Phi & = \delta_C \Phi, & \gamma_\Sigma C^\mu & = C^\lambda \nabla_\lambda C^\mu.
\end{align*}
The important feature is that if $F\in \F(\Sigma)$ is invariant under all the symmetry transformations $\del_{\rho_{\Sigma}(\xi)}$, $\xi\in\mathfrak{g}(\Sigma)$, then $\gamma_{\Sigma}F=0$, so the space of gauge invariant functionals is recovered as: $\F^{\textrm{inv}}(\Sigma)=H^0(\CE(\Sigma),\gamma_\Sigma)$.

Note that the assignment of $\g(\Sigma)$ to $\Sigma$ is a contravariant functor between $\CatSub$ and $\CatVec$ given by
\[
 (\chi^* \xi')(f) = \chi^* \xi'(\chi_* f)
\]
for $f \in \mathfrak{D}$. One can then see $\CE(\Sigma)$ as a covariant functor from $\CatSub$ to $\CatdgA$.

We started in the off-shell formalism (configurations are not required to satisfy the equations of motion), so  in the next step we would like to find a characterization of on-shell functionals. In this context, the BV complex is a very useful structure, since it allows to have a control of such quantities. Our main objective is to characterize the space of invariant on-shell functionals. Let $\F_S(\Sigma)$ denote the space of on-shell functionals. This is given by the quotient  $\F_S(\Sigma)=\F(\Sigma)/\F_0(\Sigma)$, where $\F_0(\Sigma)$ is the ideal generated by the equations of motion, i.e., by the elements $\del_\X S$, where
\[
(\partial_\X S)(\varphi) \doteq \skal{L(f)^{(1)}(\varphi)}{\X(\varphi)}, \qquad \X\in\V(\Sigma),
\]
with $f\equiv 1$ on the support of $\X$.
Let $\CE_S(\Sigma)$ denote the Chevalley-Eilenberg complex constructed for $\F_S(\Sigma)$ instead of $\F(\Sigma)$. Using the standard construction of the BV-complex \cite{RejznerFredenhagen} we obtain the resolution of $\CE_S(\Sigma)$ as
\begin{equation}
\label{eq:BV}
 \BV(\Sigma) = \mathcal{C}^\infty_{\textrm{ml}}(\E(\Sigma),\Lambda\E_c(\Sigma)\widehat{\otimes}S \mathfrak{g}_c(\Sigma)\widehat{\otimes}\Lambda\mathfrak{g}'(\Sigma)).
\end{equation}
Here $\Lambda$ denotes the antisymmetric, $S$ the symmetric tensor product, and $\widehat{\otimes}$ is the completed tensor product (for details on the topologies see \cite{RejznerFredenhagen}). The new ingredient here is  $\g_c(\Sigma)$, the space of \emph{antifields of ghosts}. It appears because $\BV(\Sigma) $ has to contain symmetric tensor powers of compactly supported derivations of the space of ghosts $\g(\Sigma)$.
Formally, a polynomial element of $\BV(\Sigma)$ can be written as
\begin{multline}
\label{Polynom2}
F=\int f_F(x_1, \dots ,y_{m})\Phi_{x_1} \dots \Phi_{x_j} C_{x_{j+1}}\wedge \dots \wedge C_{x_{k}}\Phi^\ddagger_{y_1}\wedge \dots \wedge \Phi^\ddagger_{y_l}\\ \times C^\ddagger_{y_{l+1}} \dots C^\ddagger_{y_{m}} \mu(x_1) \ldots \mu(y_{m}),
\end{multline}
where we used the notation $C^\ddagger_{y}\equiv\frac{\delta}{\delta C_{y}}$ for the antifields of ghosts.
The notion of locality or microcausality of $F$ reduces again to a condition on the support and wavefront set of the distribution $f_F$. For a definition which does not refer to the formal notation above, see  \cite{RejznerFredenhagen}.

The BV differential is defined as $s=\gamma+\delta$, where $\gamma$ is the Chevalley--Eilenberg differential\footnote{The Lie algebra $\g(\Sigma)$ acts on $\V(\Sigma)$ and $\g_c(\Sigma)$ via the commutator and one can define the Chevalley--Eilenberg differential in the similar way as in case of $\F(\Sigma)$.} and $\delta$ is the so called Koszul--Tate differential, which acts trivially on fields and ghosts, and on antifields it is given by
\begin{align*}
 \delta \X & \doteq \partial_\X S & \text{ for } \X \in \V(\Sigma), \\
 \delta \xi & \doteq \partial_{\rho(\xi)} & \text{ for } \xi\in\g_c(\Sigma).
\end{align*}
The grading of the Chevalley--Eilenberg complex is denoted by $\#\pg$ (pure ghost number) and the grading of the Koszul--Tate complex is called the antifield number ($\#\af=1$ for vector fields and $\#\af=2$ for the elements of $\g_c(\Sigma)$). The total grading is called the ghost number: $\#\gh=\#\pg-\#\af$.

On $\BV(\Sigma)$, like on $\Lambda\V(\Sigma)$, one can define the Schouten bracket (antibracket). 
To introduce it, note that the underlying algebra of the BV complex consists of graded symmetric (with respect to the ghost number) powers of derivations of  $\CE(\Sigma)$, i.e., $\BV(\Sigma)\subset \mathbf{S}\textrm{Der}(\CE(\Sigma))$, where  $  \mathbf{S}^0 \textrm{Der}(\CE(\Sigma))\doteq\CE(\Sigma)$. On such a space one can define the graded Schouten bracket $\{.,.\}$, which on $\textrm{Der}(\CE(\Sigma))$ coincides with the commutator, for a derivation and an element $F$ of $\CE(\Sigma)$ it is just the evaluation of the derivation on $F$ and for higher graded symmetric powers we extend it by the Leibniz rule. Formally, it can again be written as
\begin{equation}\label{antibracket2}
\{F,G\}=-\sum_\alpha\int\!\left(\!\frac{\delta F}{\delta\phi^\alpha(x)}\wedge\frac{\delta G}{\delta\phi_\alpha^\ddagger(x)}+(-1)^{\#\gh(F)}\frac{\delta F}{\delta\phi_\alpha^\ddagger(x)}\wedge\frac{\delta G}{\delta\phi^\alpha(x)}\!\right) \mu(x),
\end{equation}
where $\phi^\alpha=\Phi,C$ and  $\phi^\ddagger_\alpha=\Phi^\ddagger,C^\ddagger$. The advantage of introducing the antibracket is the fact that the BV differential can now be written as
\[
 s F = \{ L_{\Sigma}^\mathrm{ext}(f), F \},\qquad F\in\BV(\Sigma),
\]
where $f$ is chosen such that $f \equiv 1$ on $\supp F$, and $L^\mathrm{ext}$ is the so-called extended Lagrangian, which can again be  understood as a natural transformation between the functors $\D$ and $\BV$ (see \cite{RejznerFredenhagen}). Formally, it can be written as
\[
 L^{\mathrm{ext}}_{\Sigma}(f)(\varphi) = \lambda^{-2} \int f \frac{\sqrt{ - \tilde g}}{\sqrt{ - g }} \mu - \int f \delta_C \Phi^a \frac{\delta}{\delta \Phi^a} \mu - \frac{1}{2} \int f [C,C]^\mu \frac{\delta}{\delta C^\mu} \mu.
\]

Finally, one extends $\BV(\Sigma)$ to $\BV_\mc(\Sigma)$. This proceeds analogously to the definition of $\F_\mc(\Sigma)$, i.e., by putting the H\"ormander topology on the spaces of distributions $f_F$ occurring in \eqref{Polynom2}. Similarly, one defines the local elements as those where $f_F$ is supported on the diagonal with wave front set orthogonal to $T \Delta^{k+m}(\Sigma)$.

\section{Gauge fixing}
\label{sec:GaugeFixing}

We call the space $\Lambda \E_c(\Sigma) \hat \otimes S \g_c(\Sigma) \hat \otimes \Lambda \g'(\Sigma)$ appearing in the definition \eqref{eq:BV} of the BV complex the \emph{minimal sector}. In order to do the gauge fixing, we have to extend it by the \emph{nonminimal} sector, given by
Lagrange multipliers (elements of $S \g'(\Sigma)$, also called Nakanishi--Lautrup fields), antighosts (elements of $\Lambda \g'(\Sigma)$) and their antifields (corresponding derivations). 
Let $B^\mu_x$ denote the evaluation functional in $\g'(\Sigma)$, and $\bar{C}^\mu_x$ the evaluation functional in $\g'(\Sigma)$.
Formally, we can write the elements of $S\g'(\Sigma)$ as
\[
 F = \int f(x_1, \dots, x_k) B_{x_1}\otimes_S \dots \otimes_S B_{x_k} \mu(x_1) \dots \mu(x_k),
\]
where $\otimes_S$ stands for the symmetric tensor product, and elements of $\Lambda \g'(\Sigma)$ are represented by
\[
 G = \int f(x_1, \dots, x_k) \bar C_{x_1}\wedge \dots\wedge \bar C_{x_k} \mu(x_1) \dots \mu(x_k).
\]

The BV differential is extended to the nonminimal sector by setting $s \bar C=iB$ and $s B=0$. This definition ensures that antighosts and Lagrange multipliers form trivial pairs, so they do not modify the cohomology of $s$. The corresponding graded tensor powers of derivations are elements of $S \g_c$ (the antifields of antighosts) and $\Lambda \g_c$ (the antifields of the Lagrange multipliers). Their generators are denoted by $\frac{\delta}{\delta \bar C_x} = \bar C^\ddagger_x$ and $\frac{\delta}{\delta B_x} = B^\ddagger_x$.
The non-minimally extended BV-complex is now given by
\[
 \BV^\nm = \mathcal{C}^\infty_\ml(\E,(\textrm{Minimal sector})\hat \otimes \Lambda \g' \hat \otimes S \g' \hat \otimes S \g_c \hat \otimes \Lambda \g_c).
\]
The space $\BV^\nm(\Sigma)$ can again be equipped with the Schouten bracket (antibracket), which formally can be written as (\ref{antibracket2}), now with $\phi^\alpha=\Phi,C,\bar C, B$ and  $\phi^\ddagger_\alpha=\Phi^\ddagger,C^\ddagger,\bar C^\ddagger, B^\ddagger$.
The BV differential $s$ can be written locally as an antibracket with the extended action $S^{\mathrm{ext}}$, where now
\begin{multline*}
 L^{\mathrm{ext}}_{\Sigma}(f)(\varphi) = \lambda^{-2} \int f \frac{\sqrt{ - \tilde g}}{\sqrt{ - g }}  \mu - \int f \delta_C \varphi^a \frac{\delta}{\delta \Phi^a} \mu \\
 - \frac{1}{2} \int f [C,C]^\mu \frac{\delta}{\delta C^\mu}  \mu - i \int f B^\mu \frac{\delta}{\delta \bar C^\mu}  \mu\,.
\end{multline*}

The benefit of the extension by the nonminimal sector is that one can now perform a gauge fixing. Assume that a set of local constraints $T^\alpha_x(\varphi(x)) = 0$ leads to a hyperbolic equation of motion for $\varphi$. One defines the corresponding \emph{gauge fixing fermion}\footnote{Sometimes a term $\int \frac{\alpha}{2} \bar C_\mu B^\mu \mu$ is added. Omitting it corresponds to the Landau gauge in Yang--Mills theories.} 
\[
 \Psi(f)(\varphi) = \int f \bar C_\alpha T^\alpha(\varphi) \mu,
\]
and performs the canonical transformation
\[
 L^\Psi_\Sigma(f) = \sum_{k = 0}^\infty \frac{1}{k!} \{ \dots \{ L^\mathrm{ext}_\Sigma(f), \underbrace{ \Psi(f') \}, \dots , \Psi(f')}_k \},
\]
with $f' \equiv 1$ on $\supp f$, so that the \rhs does not depend on the choice of $f'$, by locality. In the present case, $L^\mathrm{ext}$ is of first order in the antifields, so the series truncates after $k = 1$. Considering only the part of the Lagrangean with $\#\af = 0$, and linearizing, one obtains a hyperbolic wave operator $S_0''$.

The next goal is to identify a suitable set of constraints, i.e., a gauge fixing fermion. For simplicity, we will from now on restrict the target space $M$ to be the $n$-dimensional Minkowski space. In Cartesian coordinates, the infinitesimal change under reparametrizations is then given by
\[
 \delta_\xi \varphi^a = \xi^\mu \ud \tilde X^a_\mu = \xi^\mu \left( \ud X^a_\mu + \lambda \del_\mu \varphi^a \right),
\]
for $\xi \in \g$. Furthermore, the induced metric is
\[
 \tilde g_{\mu \nu} = g_{\mu \nu} + \lambda \ud X^a_\mu h_{ab} \del_\nu \varphi^b + \lambda \ud X^a_\nu h_{ab} \del_\mu \varphi^b + \lambda^2 \del_\mu \varphi^a h_{ab} \del_\nu \varphi^b. 
\]
For the equation of motion stemming from the Nambu--Goto Lagrangean \eqref{eq:Nambu-Goto} we obtain
\begin{equation}
\label{eq:eom0}
 0 = \del_\mu \left( \tilde g^{\mu \nu} \sqrt{- \tilde g} \ud \tilde X_\nu^a \right) = \nabla_\mu \left( \tilde g^{\mu \nu} \sqrt{- \tilde g} \ud \tilde X_\nu^a \right).
\end{equation}
The second equation follows from $\del_\mu f^\mu = \nabla_\mu f^\mu$, which is true for any densitized vector $f$.

We recapitulate the ingredients introduced so far in the following table:
\begin{center}
\begin{tabular}{|c||c|c|c|c|c|}
\hline
 & $\#\af$ & $\#\pg$ & $\#\gh$ & action of $\gamma$ & action of $\delta$ \\
\hline
\hline
 $\Phi^a$ & 0 & 0 & 0 & $C^\mu \ud \tilde X^a_\mu$ & 0 \\
%\hline 
 $C^\mu$ & 0 & 1 & 1 & $C^\lambda \nabla_\lambda C^\mu$ & 0 \\
%\hline
 $\bar C^\mu$ & 0 & -1 & -1 & $iB^\mu$ & 0 \\
%\hline
 $B^\mu$ & 0 & 0 & 0 & 0 & 0 \\
\hline
 $\Phi^\ddagger_a$ & 1 & 0 & -1 & $\nabla_\mu (C^\mu \Phi^\ddagger_a)$  & $- h_{ab} \nabla_\mu (\tilde g^{\mu \nu} \sqrt{-\tilde g} \ud \tilde X^b_\nu)$ \\ 
%\hline
 $C^\ddagger_\mu$ & 2 & 0 & -2 &  $\nabla_\lambda( C^\lambda C^\ddagger_\mu) + \nabla_\mu C^\lambda C^\ddagger_\lambda$ & $\ud \tilde X^a_\mu \Phi^\ddagger_a$ \\
%\hline
 $\bar C^\ddagger_\mu$ & 0 & 0 & 0 & 0 & 0 \\
%\hline
 $B^\ddagger_\mu$ & 1 & 0 & -1 & 0 & $- i \bar C^\ddagger_\mu$ \\
\hline
\end{tabular}
\end{center}
Note that these explicit formulas are valid only on Minkowski space, and that $\tilde X$ and the derived expressions $\tilde g^{\mu \nu}$ and $\tilde g$ are obtained by replacing $\varphi^a$ with the evaluation functional $\Phi^a$.

The equation of motion \eqref{eq:eom0} is not hyperbolic. A natural gauge condition is
\begin{equation}
\label{eq:harmonicGauge}
 \nabla_\mu  \left( \tilde g^{\mu \nu} \sqrt{- \tilde g} \right) = 0,
\end{equation}
which corresponds to the harmonic background gauge used in linearized gravity. Note that we use the covariant derivative $\nabla_\mu$ \wrt the background metric $g$. This is done in order to ensure covariance. In this gauge, one has the equation of motion
\[
 \tilde g^{\mu \nu} \nabla_\mu \ud \tilde X_\nu^a = 0.
\]

We also remark that \eqref{eq:harmonicGauge} is a vector condition, in accordance to choosing Lagrange multipliers from $\g'$.
%\marginpar{Changed sentence}
It is well known that the harmonic gauge can always be achieved, at least locally \cite{WaldGR}.
%The condition \eqref{eq:harmonicGauge} corresponds to $d$ constraints. As a reparametrization is given by a vector field, i.e., $d$ smooth functions, this gauge fixing can always be achieved, at least locally.

The gauge fixing fermion for the gauge condition \eqref{eq:harmonicGauge} is
\[
 \Psi(f)(\varphi) = - i \lambda^{-2} \int f \bar C_{\nu} \frac{1}{\sqrt{-g}} \nabla_\mu \left( \tilde g^{\mu \nu} \sqrt{- \tilde g} \right) \mu.
\]
Note that $\bar C$, and hence also $B$, has then the dimension of a length.
For the gauge fixed Lagrangean at antifield number 0, one thus obtains
\begin{align*}
 L^\Psi_{\Sigma}(f)(\ph) & = \lambda^{-2} \int f  \frac{\sqrt{-\tilde g}}{\sqrt{-g}} \mu + \lambda^{-2} \int f B_{\nu} \frac{1}{\sqrt{-g}} \nabla_\mu \left( \tilde g^{\mu \nu} \sqrt{- \tilde g} \right) \mu \\
 & + i \lambda^{-2} \int f \nabla_{(\nu} \bar C_{\mu)} \left( 2 \sqrt{-\tilde g} \tilde g^{\lambda \nu} \nabla_\lambda C^\mu - \nabla_\lambda \left( \tilde g^{\mu \nu} \sqrt{-\tilde g} C^\lambda \right) \right) \frac{1}{\sqrt{-g}} \mu.
\end{align*}
Note that this is consistent with $C$ having the dimension of a length. For the perturbative expansion, we replace $C, \bar C$, and $B$ by $\lambda C$, $\lambda \bar C$, and $\lambda B$. This entails that now the mass dimension of these fields is $d/2-1$.
Expanding the Lagrangean in $\lambda$, and taking into account that the background is on-shell, we obtain a constant term of $\order(\lambda^{-2})$, a vanishing term at $\order(\lambda^{-1})$, and a free Lagrangean at $\order(\lambda^0)$, whose antifield-independent part is
\begin{subequations}
\label{eq:L_0}
\begin{align}
\label{eq:L_0_B}
 L^\Psi_{0, \Sigma}(f) & = \int f \tfrac{1}{2} (Q \Phi)^a \left( h_{ab} \nabla^\mu \del_\mu - 2 h_{ac} h_{bd} \nabla^\nu \ud X^c_\mu \ud X^d_\nu \del^\mu \right) (Q \Phi)^b  \mu \\
\label{eq:L_0_C}
 & - \int f \left[  B^\nu \ud X^a_\nu h_{ab} \nabla^\mu \del_\mu \Phi^b + i \bar C_\mu \nabla_\nu \nabla^\nu C^\mu + i \bar C_\mu R^\mu_\nu C^\nu  \right] \mu,
\end{align}
where $Q$ projects on the normal bundle $N_\Sigma M$,
\begin{align*}
P^a_b & \doteq \ud X^a_\mu g^{\mu \nu} \ud X^c_\nu h_{bc} & Q & = \id - P.
\end{align*}
The corresponding free BRST transformation $\gamma_0$ is then
\begin{align*}
 \gamma_0 \Phi^\alpha & = C^\mu \ud X^a_\mu, \\
 \gamma_0 C^\mu & = 0, \\
 \gamma_0 \bar C^\mu & = i B^\mu, \\
 \gamma_0 B^\mu & = 0.
\end{align*}
\end{subequations}
The wave operator $S_0''$ corresponding to the free action \eqref{eq:L_0} is block diagonal \wrt the ghost sector $(C^\mu, \bar C_\nu)$ and the sector $(B_\mu, (P \Phi)^a, (Q \Phi)^a)$. In the latter, it is of the form
\begin{equation}
\label{eq:S_0''}
 S_0'' = - \begin{pmatrix}
 0 & \ud X^c_\mu \nabla^\nu \del_\nu h_{cd} P^d_b  & \ud X^c_\mu \nabla^\nu \del_\nu h_{cd} Q^d_b \\
 & 0 & 0 \\
 & & O^a_b
 \end{pmatrix},
\end{equation}
where $O^a_b$ abbreviates
\[
 O^a_b = Q^a_c (\delta^{c}_{e} \nabla^\nu \del_\nu - 2 (\nabla^\rho \ud X^c_\lambda) \ud X^d_\rho h_{de} \del^\lambda) Q^e_b.
\]
We omitted the lower triangular part for simplicity. The principal part of this operator is
\[
 - \begin{pmatrix}
 0 & \ud X^c_\mu  h_{cb}  & 0 \\
 \ud X^a_\nu & 0 & 0 \\
 0 & 0 & \delta^a_b
 \end{pmatrix} \nabla^\lambda \del_\lambda,
\] 
so up to the multiplication with an invertible operator, the principal symbol coincides with the induced metric. The same holds for the ghost sector. It follows that, up to an invertible linear transformation, the wave operator $S_0''$ is normally hyperbolic \cite{BGP07}, and hence has a well-posed Cauchy problem.

\begin{remark}
\label{rem:bc}
We may now relax the condition of global hyperbolicity of $\Sigma$ by the weaker assumption that, given suitable boundary conditions at  timelike (i.e., not spacelike or null) boundaries, the wave operator $S_0''$ has unique retarded and advanced propagators. The boundary conditions should be linear in the configurations and compatible with the free BRST transformation $\gamma_0$ in the sense that it maps configurations fulfilling them to configurations which also satisfy them. The category morphisms are required to respect the boundary conditions, i.e., if $\Sigma \subset \Sigma'$, and $\Sigma$ and $\Sigma'$ share a timelike boundary, then there is a morphism $\chi: \Sigma \to \Sigma'$ only if the boundary conditions coincide at the common timelike boundary.
\end{remark}

%\marginpar{New remark}
\begin{remark}
\label{rem:Kleinert}
An alternative gauge condition, proposed in \cite{Kleinert}, would be to allow only for transversal fluctuations, i.e., to require $P \ph = 0$. This gauge condition does not strictly fit into the present framework, as the resulting wave operator is not normally hyperbolic, being of $0$th order for $P \ph$ and the auxiliary fields. Nevertheless, it is invertible. Dividing out the free equations of motion, this boils down to setting $P \ph$ and the auxiliary fields to $0$. For a plane background, this leads to the L\"uscher--Weisz string. A problem with this gauge is that it can in general not be achieved, even locally. We plan to address the relation between the two gauges in the future.
\end{remark}

\subsection{The Peierls bracket}
\label{sec:Peierls}
The free Peierls bracket is defined via the fundamental solution $\Delta$ corresponding to the wave operator $S_0''$, i.e., the difference of retarded and advanced propagator, by
\begin{equation*}
\Pei{F}{G} \doteq \sum_{\al,\beta}(-1)^{(1+ \betrag{F}) \betrag{\phi^\al} } \skal{\frac{\delta F}{\delta\phi^\al}}{{\De}^{\al\beta}\frac{\delta G}{\delta\phi^\beta}},
\end{equation*}
where $\phi^\al = \Phi, C, \bar C, B$, $|F| \doteq \#\gh(F)$ and both $F$ and $G$ have $\#\af=0$.\footnote{The Peierls bracket acts trivially on antifields in the sense that the part $\ph^\ddagger_\alpha$ of $\X = \int f_{\X}^\alpha \ph^\ddagger_\alpha \mu$ commutes through the bracket, which then only acts on $f_\X$.}

%By inspection of $S_0''$, one easily sees 
From the block-diagonality of $S_0''$ it follows that the Peierls brackets $\Pei{C}{B}$, $\Pei{\bar C}{B}$, $\Pei{B}{B}$, $\Pei{C}{\varphi}$, and $\Pei{\bar C}{\varphi}$ vanish. Furthermore,
\begin{equation*}
 \Pei{C^\mu(x)}{\bar C^\nu(y)} = - i \Delta^{\mu \nu} (x,y),
\end{equation*}
where $\Delta^{\mu \nu}$ is the fundamental solution for the wave operator $-\nabla^\mu \nabla_\mu - \mathrm{Ric}$ on vector fields. 
%In order to discuss the Peierls brackets involving $B$ and $\Phi$, we introduce the notation $Q \doteq \1 - P$, with $P$ given by \eqref{eq:P}. Then $Q$ is a projector onto the normal bundle, and we use $P$ and $Q$ to split $TM = T \Sigma \oplus N \Sigma$. In the basis $(B_\mu, (P \Phi)^a, (Q \Phi)^a )$, we find that the operator $S_0''$ is of the form\footnote{That $\delta S_0 / \delta (P \Phi)$ is independent of $\Phi$ is easily seen by contracting \eqref{eq:eom_varphi_1} from the left with $P$. This reflects the gauge invariance of the Nambu--Goto Lagrangean.}
%\begin{equation}
%\label{eq:S_0''}
% S_0'' = - \begin{pmatrix}
% 0 & \ud X^c_\mu \nabla^\nu \nabla_\nu h_{cd} P^d_b  & \ud X^c_\mu \nabla^\nu \nabla_\nu h_{cd} Q^d_b \\
% & 0 & 0 \\
% & & O^a_b
% \end{pmatrix},
%\end{equation}
%where $O^a_b$ abbreviates
%\[
% O^a_b = Q^a_c (\delta^{c}_{e} \nabla^\nu \del_\nu - 2 (\nabla^\rho \ud X^c_\lambda) \ud X^d_\rho h_{de} \del^\lambda) Q^e_b.
%\]
%We omitted the lower triangular part for simplicity. Note that the operator in the upper right (and lower left) corner is of first order. The corresponding fundamental solution is of the form
The fundamental solution corresponding to the wave operator \eqref{eq:S_0''} in the $(B, P\Phi, Q\Phi)$ sector is of the form
\begin{equation}
\label{eq:Delta}
 \Delta = \begin{pmatrix} 0 & \Delta_{\mu b} & 0 \\
 \Delta^{a \nu} & \tilde \Delta^a_b & \hat \Delta^a_b \\
 0 & \check \Delta^a_b& \bar \Delta^a_b \end{pmatrix}.
\end{equation}
Closer inspection of $S_0''$ shows that $\Delta^{a \nu}(x,y) = \ud X^a_\mu(x) \Delta^{\mu \nu}(x,y)$, and correspondingly for $\Delta_{\mu b}$. Furthermore, $\bar \Delta$ is the fundamental solution to the normally hyperbolic operator $O$ on $N \Sigma$. For our purposes, we do not have to explicitly determine $\tilde \Delta$, $\hat \Delta$, and $\check \Delta$.

\subsection{The BRST current}
For later reference, we compute the free BRST current. It is defined as
\begin{multline*}
 j_0^\mu = \sum_\al  \left[ \gamma_0 \phi^\al \frac{\delta}{\delta \nabla_\mu \phi^\al} L_0 + 2 \nabla_\nu \gamma_0 \phi^\al \frac{\delta}{\delta \nabla_\mu \nabla_\nu \phi^\al} L_0 \right. \\ \left. - \nabla_\nu \left( \gamma_0 \phi^\al \frac{\delta}{\delta \nabla_\mu \nabla_\nu \phi^\al} L_0 \right) \right] - J_0^\mu,
\end{multline*}
where $J^\mu_0$ is the divergence term and $\phi^\al = \Phi, C, \bar C, B$. First of all, we note that the current corresponding to the Nambu--Goto Lagrangean \eqref{eq:Nambu-Goto} vanishes (it is cancelled by $J_0^\mu$). The same is thus true for the terms in the free Lagrangean $L_0$ that stem from its expansion. Hence, we only need to consider the terms involving the auxiliary fields. For the action of $\gamma_0$ on these, we obtain
\[
 \gamma_0 L_{0, \text{auxiliary part}} = - \nabla_\mu \left( B_\nu \nabla^\mu C^\nu + B_\nu \nabla^\nu C^\mu - B^\mu \nabla_\nu C^\nu \right)
\]
and thus
\[
 J_0^\mu = - B_\nu \nabla^\mu C^\nu - B_\nu \nabla^\nu C^\mu + B^\mu \nabla_\nu C^\nu.
\]
With the above, we obtain
\begin{equation}
\label{eq:current}
 j_0^\mu = \nabla^\mu B_\nu C^\nu - B_\nu \nabla^\mu C^\nu.
\end{equation}
Obviously, it is covariantly conserved, up to terms of $\order(\lambda)$.
As no renormalization ambiguities in the Wick powers $BC$ occur (by the discussion in the previous subsection, $B$ and $C$ commute), the same will be true in the quantized theory.

\section{Fields}
\label{sec:Fields}
Let us now discuss possible observables for the model. A natural class of observables would be the following: Take $f \in C_c^\infty(M)$ and define
\begin{equation}
\label{eq:Observable}
 \Psi_{\Sigma}(f)(\varphi) = \int f(\tilde X(x)) \frac{\sqrt{-\tilde g}}{\sqrt{-g}} \mu(x).
\end{equation}
By construction, this is invariant under reparametrizations.
In a nonperturbative setting, it would even be invariant under arbitrary changes in the background $\Sigma$. For $d=1$, this gives the eigentime the particle spends in the region described by $f$.
%Even though such observables have a nice physical interpretation, they are not in $\mathfrak{F}(\Sigma)$, as they do not have compact support on $\Sigma$. It would thus be necessary to introduce a localization function $g$ on $\Sigma$ and require that it transforms in the correct way in order to ensure reparametrization invariance. This leads to the framework of natural transformations discussed below.
As we are in a perturbative setting, \eqref{eq:Observable} should rather be written as\footnote{For simplicity, we here assume that $M$ is flat.}
\begin{equation}
\label{eq:ObservableExpanded}
 \Psi_{\Sigma}(f)(\varphi) = \sum_{k=0}^\infty \frac{\lambda^k}{k!} \int \ph^{a_1}(x) \dots \ph^{a_k}(x) (\del_{a_1} \dots \del_{a_k} f)(X(x)) \frac{\sqrt{-\tilde g}}{\sqrt{-g}} \mu(x).
\end{equation}
In order for this to be in $\mathfrak{F}(\Sigma)$, we have to require that $\Sigma \cap \supp f$ is compact. One can eliminate the test function on the target space, replacing it by a symmetric test section $t \in \Gamma_c^\infty(\Sigma, T^\otimes M)$, with
\[
 T^\otimes M = \bigoplus_{k} \underbrace{T^* M \otimes \dots \otimes T^* M}_k,
\]
by substituting the component $t_{a_1 \dots a_k}(x)$ of rank $k$ for $(\del_{a_1} \dots \del_{a_k} f)(X(x))$ in \eqref{eq:ObservableExpanded}. One only has to impose the following consistency condition between the components of different rank:
\begin{equation}
\label{eq:tRelation}
 \del_\mu t_{a_1 \dots a_k} = \ud X^a_\mu t_{a a_1 \dots a_k}.
\end{equation}
It is straightforward to check that the functionals obtained in this way are gauge invariant. Hence, there is a large class of localized observables, despite the fact that the model has diffeomorphism invariance. Denoting by $\Tens_c'(\Sigma)$ the vector space of test tensors described above, the assignment $\Sigma \mapsto \Tens_c'(\Sigma)$ is a covariant functor from $\CatSub$ to $\CatVec_i$, analogously to $\E_c$ and $\D$. The assignment of a functional $\Psi(t)$ to a test tensor $t$ then defines a natural transformation between $\Tens_c'$ and $\BV_\loc^\nm$, and thus a field in the sense defined in \cite{BrunettiFredenhagenVerch}.

%We now have the following:
%\begin{proposition}
%The non-trivial component of \eqref{eq:ObservableExpanded} of lowest order in $\lambda$ is either a constant or only involves the transversal fluctuations $Q \ph$.
%\end{proposition}
%\begin{proof}
%Assume that the lowest nontrivial order in $\lambda$ is $p > 1$. Then it may be sensitive to longitudinal fluctuations $P \ph$ only if $\ud X^a_\mu t_{a a_2 \dots a_p} \neq 0$. Hence, by \eqref{eq:tRelation}, $t_{a_1 \dots a_{p-1}} \neq 0$, so already the component of $\order(\lambda^{p-1})$ is nontrivial. Now let $p=1$. At this order, two terms contribute to \eqref{eq:ObservableExpanded}:
%\[
% \int \ph^a \del_a f(X(x)) \sqrt{-g} \ud x + \int f(X(x)) \del_\mu \ph^a \ud X^b_\nu g^{\mu \nu} h_{ab} \sqrt{-g} \ud x.
%\]
%\end{proof}

One can also be more general and define observables analogously to the general relativity case \cite{RejznerFredenhagen}, where one assumes that the test tensors also transform under gauge transformations.
One defines fields as natural transformations from tensor powers of tensor test sections to the BV-complex. Here, it is natural to use the covariant functor $\Tens_c$ between $\CatSub$ and $\CatVec_i$, which associates to each submanifold $\Sigma$ the vector space $\Tens_c(\Sigma) = \Gamma^\infty_c(\Sigma, T^\otimes M \otimes T^\otimes \Sigma)$.
Gauge transformations then also act on the test section, by an action $\rho$ of $\mathfrak{g}$ on $\Tens_c$, which acts by the Lie derivative on the $T^\otimes \Sigma$ indices and by the covariant derivative of $M$ on the $T^\otimes M$ indices.
%We introduce the physical fields as in \cite{RejznerFredenhagen}. There, fields are defined as natural transformations from tensor powers of tensor test sections to the BV-complex. Here, it is natural to use the covariant functor $\Tens_c$ between $\CatSub$ and $\CatVec_i$, which associates to each submanifold $\Sigma$ the vector space $\Tens_c(\Sigma) = \Gamma^\infty_c(\Sigma, T^\otimes M \otimes T^\otimes \Sigma))$, where
%\[
% T^\otimes M = \bigoplus_{k,l} \underbrace{T M \otimes \dots \otimes T M}_k \otimes \underbrace{T^* M \otimes \dots \otimes T^* M}_l
%\]
%and analogously for $T^\otimes \Sigma$.
%On $\mathfrak{Tens}_c$, we also define an action $\rho$ of $\mathfrak{g}$, which acts by the Lie derivative on the $T^\otimes \Sigma$ indices and by the covariant derivative of $M$ on the $T^\otimes M$ indices.
Fields are now defined as
\[
 Fld = \bigoplus_{k = 0}^\infty \mathrm{Nat}(\mathfrak{Tens}_c^k, \BV^\nm).
\]
The algebra of physical fields is then given by
\begin{equation}
\label{eq:physFields}
 Fld_{\mathrm{ph}} = H^0(Fld, s),
\end{equation}
where $s$ acts on elements of $Fld$ by
\begin{equation}
\label{eq:gammaFields}
 (s \Phi)_{\Sigma}(t) = s \Phi_{\Sigma}(t) + (-1)^\betrag{\Phi} \Phi_{\Sigma}(\rho_{\Sigma}(\cdot) t).
\end{equation}
In the first term on the r.h.s., $s$ is the action on functionals given by the action of $\gamma$ and $\delta$ on $\Phi, C, \bar C, B$ specified in the table in Section~\ref{sec:GaugeFixing} for a flat target space. The placeholder in the second term on the \rhs takes an elements of $\Gamma^\infty(\Sigma, T \Sigma)$, i.e., it is a ghost. As discussed in \cite{RejznerFredenhagen}, this can be obtained as the Chevalley--Eilenberg differential of an action of $\g$ on $Fld$. 
An example for an observable in $Fld_\mathrm{ph}$ would be, for $t \in \Gamma_c^\infty(\Sigma, T^* M)$,
\begin{equation}
\label{eq:yField}
 \Phi_\Sigma(t)(\varphi) = \int t_a \tilde X^a \frac{\sqrt{-\tilde g}}{\sqrt{-g}} \mu.
\end{equation}
Note that due to the presence of $\sqrt{-\tilde g}$, this is not linear in $\ph$. The absence of linear fields is a phenomenon already encountered in general relativity \cite{RejznerFredenhagen}.
Also note that in principle one should consider equivalence classes of $t$ corresponding to the equations of motion, but due to the nonlinearity, these are hard to characterize directly.

The physical fields do not correspond to the positive definite subspace in typical representations. This can be fixed if we restrict to observables that respect the perturbative expansion in the following sense: For a given background $\Sigma$, expand $\Phi_\Sigma$ in powers of $\lambda$,
\[
 \Phi_\Sigma = \sum_{k = 0}^\infty \lambda^{k} \Phi_{\Sigma, k},
\]
where $\Phi$ is a physical field.
Given a test tensor $t$, denote by $p_t$ the lowest integer such that $\Phi_{\Sigma, p_t}(t) \neq 0$. The expression $\rho_\Sigma(\cdot)$ in the second term on the \rhs of \eqref{eq:gammaFields} is of $\order(\lambda)$, due to the additional ghost, so we require that $\Phi_{\Sigma}(\rho_\Sigma(\cdot) t)$ is of $\order(\lambda^{p_t+1})$. An example where this is violated is the field \eqref{eq:yField} on the flat background \eqref{eq:flatBackground} with $t_a \neq 0$ for $a\in \{0,1\}$ and $\int t_a X^a \mu = 0$. Due to the latter condition, $p_t = 1$, but $\Phi_{\Sigma}(\rho_\Sigma(\cdot) t)$ is of $\order(\lambda)$. Assuming that the condition is fulfilled, we have
\[
 (s \Phi)_\Sigma(t) = \lambda^{p_t} s_0 \Phi_{\Sigma, p_t}(t) + \order(\lambda^{p_t+1}),
\]
where $s_0$ is the free part of the BV differential $s$. It follows that $s_0 \Phi_{\Sigma, p_t}(t) = 0$, so $\Phi_{\Sigma, p_t}(t)$ is, at the linearized level, a physical observable in the usual sense. Thus, in a representation in which the cohomology of the BV operator defines a positive definite subspace, the positivity (in the sense of formal power series in $\lambda$) is ensured.
In the flat background case \eqref{eq:flatBackground}, examples of such observables are \eqref{eq:yField} with $t$ being normal to $\Sigma$.
As will be shown below, the physical excitations are precisely the transversal ones, so there is a one--to--one correspondence between these and the described observables.

Certainly, the relation between the physical fields as defined by \eqref{eq:physFields} and the physical excitations as defined by the cohomology of $s$ on a fixed background deserves a more detailed study. Preliminary results suggest that it is possible to define a modified BV operator whose cohomology, together with a condition on the perturbative expansion, defines physical states that are compatible with the physical fields as defined here\footnote{Klaus Fredenhagen and Katarzyna Rejzner, work in progress.}.

\section{Quantization}
\label{sec:Quantization}
\subsection{General structure}
We quantize the theory in the sense of deformation quantization using ideas from causal perturbation theory. Therefore, we will first quantize the free (linearized) theory and then define the interacting fields by the formula of Bogoliubov \cite{BS}. This approach, called perturbative algebraic quantum field theory, is presented in \cite{BDF09} and references therein, on the example of the scalar field. Recently, also Yang--Mills theory \cite{HollandsYM} and general gauge theories including gravity \cite{RejznerFredenhagenQuantization,Rej11} were treated in this framework. 

To illustrate the main ideas of this approach we consider first the space $\BV^\nm_\reg(\Sigma)$ of regular elements of the nonminimally extended BV complex.  Here we call an element of the form \eqref{Polynom2} regular if $f_F$ is a smooth compactly supported section. The space of regular gauge invariant on-shell functionals is the 0-th cohomology of the BV differential on $\BV^\nm_\reg(\Sigma)$. Moreover $\BV^\nm_\reg(\Sigma)$ forms 
a Poisson algebra with the pointwise product $m:  F\otimes  G  \mapsto  F\cdot G$ and with the Peierls bracket $\Pei{\cdot}{\cdot}$ introduced in Section~\ref{sec:Peierls} as the Poisson bracket.
In the framework of deformation quantization \cite{BDF09}, the observables of the quantized theory are constructed as formal power series in $\hbar$ with coefficients in $\BV^\nm_{\reg}(\Sigma)$. On these, one defines an associative noncommutative product $\star$ such that for $\hbar\rightarrow  0$, $F \star G \rightarrow F \cdot G$ and $[F,G]_\star/i\hbar\rightarrow \Pei{F}{G}$. 

For the Poisson algebra of functions on a finite dimensional Poisson manifold, the deformation quantization exists in the sense of formal power series due to a theorem of Kontsevich \cite{Kontsevich}. In field theory, the formulas of Kontsevich lead to ill defined terms, and a general solution of the problem is not known. Nevertheless, for the quadratic part of the action one can define the $\star$-product directly by means of
\begin{equation*}
F\star G\doteq m\circ \exp({i\hbar \DCp})(F\otimes G),
\end{equation*}
where  $\DCp$ is the functional differential operator
\begin{equation*}
\DCp(F \otimes G) \doteq \frac{1}{2} \sum_{\al, \beta} (-1)^{(1+\betrag{F}) \betrag{\phi^\al}} \int {\De}^{\al\beta}(x,y) \frac{\delta F}{\delta\phi^\al(x)} \otimes \frac{\delta G}{\delta\phi^\beta(y)} \mu(x) \mu(y),
\end{equation*}
with $\Delta^{\alpha \beta}$ the fundamental solution corresponding to $S_0''$.
The complex conjugation satisfies the relation $\overline{F\star G}=\overline{G}\star\overline{F}$, therefore we can use it to define an involution  $F^*(\ph)\doteq\overline{F(\ph)}$.
The resulting involutive topological algebra $\A_\reg(\Sigma) \doteq (\BV^\nm_\reg(\Sigma )[[\hbar]], \star, *)$ is the quantization of $(\BV^\nm_\reg(\Sigma ), \Pei{\cdot}{\cdot})$. 

In the next step we want to introduce the interaction. To this end we consider another product on $\A_\reg(\Sigma)$, namely the time-ordered product, defined as
\begin{equation*}
F \T G \doteq m \circ \exp(i\hbar \DDp)(F \otimes G),
\end{equation*}
with the functional differential operator
\begin{equation*}
\DDp(F \otimes G) \doteq \sum_{\al,\beta} (-1)^{(1+\betrag{F}) \betrag{\phi^\al}} \int {\De_D}^{\al, \beta}(x,y) \frac{\delta F}{\delta\phi^\al(x)} \otimes \frac{\delta G}{\delta\phi^\beta(y)} \mu(x) \mu(y),
\end{equation*}
where $\De_D=\frac{i}{2}(\De_A+\De_R)$ is called the Dirac propagator. Due to the support properties of the propagators, $\T$ coincides with $\star$ for functionals with time ordered supports. Moreover, it is equivalent in the sense of $\star$-products to the pointwise product of classical field theory by the so-called time ordering operator $\TT$, defined as
\[
\TT(F)\doteq e^{i\hbar\DD}(F),
\]
with 
\begin{equation*}
\DD(F) \doteq \sum_{\al,\beta} (-1)^{(1+\betrag{F}) \betrag{\phi^\al}} \int {\De_D}^{\al, \beta}(x,y) \frac{\delta^2 F}{\delta\phi^\al(x)\delta\phi^\beta(y)} \mu(x) \mu(y),
\end{equation*}
by the formula
\begin{equation*}
F \T G \doteq \TT(\TT^{-1} F \cdot \TT^{-1}G).
\end{equation*}
Here $\TT^{-1}$ is the inverse of $\TT$ in the sense of formal power series. Note that $\phi^\al = \Phi, C, \bar C, B$, so neither $\star$ nor $\T$ affect the antifields $\phi^\ddagger_\al$.
We denote by $S_0$ the antifield-independent part of the free gauge-fixed action \eqref{eq:L_0}. The remainder
\[
 S_\ia = S_0 - S_\ext = S_1 + \theta_0 + \theta_1
\]
is then the interacting part. For now we assume that it is an element of $\A_\reg(\Sigma)[[\lambda]]$. The formal S-matrix is defined as the time-ordered exponential
\begin{equation}\label{Smatrix}
\Scal(S_\ia) \doteq \expT{S_\ia} = \TT(e^{\TT^{-1} S_\ia}) = \sum_{k=0}^\infty \frac{1}{k!} \underbrace{S_\ia \T \dots \T S_\ia}_k.
\end{equation}
We can now define the relative S-matrix for $F\in\A_\reg(\Sigma)$, $S_\ia \in \A_\reg(\Sigma)[[\lambda]]$ by the formula of Bogoliubov,
\begin{equation*}%\label{Bog}
\Scal_{S_\ia}(F) \doteq \Scal(S_\ia)^{\star-1} \star \Scal(S_\ia+F).
\end{equation*}
Interacting quantum fields are generated by $\Scal_{i S_\ia / \hbar}(F)$ and can be written as
\begin{equation}\label{Rv}
 R_{S_\ia}(F) \doteq \frac{\ud}{\ud \eta}\Big|_{\eta=0}\Scal_{i S_\ia / \hbar}(\eta F) = \left( \expT{i S_\ia / \hbar} \right)^{\star-1} \star \left( \expT{i S_\ia / \hbar} \T F \right).
\end{equation}
This is to be understood in the sense of formal power series in $\hbar$, $\lambda$, and the antifields in $S_\ia$. Note that the terms with negative powers of $\hbar$ on the \rhs cancel.

The algebra $\A_\reg(\Sigma)$ is equipped with the time ordered Schouten bracket $\{.,.\}_{\sst{\TT}}$ defined as
\begin{equation*}
\{ F, G \}_{\sst{\TT}}=\mathcal{T}\{\mathcal{T}^{-1}F,\mathcal{T}^{-1}G\}.
\end{equation*}
More explicitly,
\begin{equation*}
\{F,G\}_{\sst{\TT}}=-\sum_\alpha\int\!\left(\!\frac{\delta F}{\delta\phi^\alpha(x)}\T\frac{\delta G}{\delta\phi_\alpha^\ddagger(x)}+(-1)^{|F|}\frac{\delta F}{\delta\phi_\alpha^\ddagger(x)}\T\frac{\delta G}{\delta\phi^\alpha(x)}\!\right) \mu(x),
\end{equation*}
where $\betrag{F} = \#\gh(F)$. The classical ideal generated by the free equations of motion is transformed into the image of the time-ordered Koszul operator
\begin{equation*}
\delTo = \{ \cdot, S_0\}_{\sst{\TT}}.
\end{equation*}

Now we want to characterize the quantum ideal of the equations of motion. Following \cite{RejznerFredenhagenQuantization} we define\footnote{Note that this is \textit{not} a Poisson bracket, as $\star$ is not graded commutative.}
\begin{equation*}
\{F,G\}_{\star}=-\sum_\alpha\int\!\left(\!\frac{\delta F}{\delta\phi^\alpha(x)}\star\frac{\delta G}{\delta\phi_\alpha^\ddagger(x)}+(-1)^{|F|}\frac{\delta F}{\delta\phi_\alpha^\ddagger(x)}\star\frac{\delta G}{\delta\phi^\alpha(x)}\!\right) \mu(x).
\end{equation*}
The quantum ideal of equations of motion is characterized as the image of $\{.,S_0\}_\star$. It is related to the classical one by
\begin{equation*}
i \hbar \Lap(F)=\{F,S_0\}_{\sst\TT}-\{F,S_0\}_\star,
\end{equation*}
where $\Lap$ is a functional differential operator defined on $\A_\reg(\Sigma)$ by
\[
\Lap F=\sum\limits_\alpha(-1)^{\betrag{\phi_\al} (1+\betrag{F})} \int \frac{\delta^2 F}{\delta\phi_\al^\ddagger(x)\delta\phi^\al(x)} \mu(x).
\]
The main difficulty in quantum BV formalism is that $\Lap$ in the above form is not well defined on local elements of the BV complex. However, it was shown in \cite{RejznerFredenhagenQuantization} that this problem can be resolved if we replace in all the definitions the non-renormalized time-ordered product $\T$ with the renormalized one $\TR$. We will come back to this discussion in Section \ref{renorm}.

Finally, we discuss the quantum master equation and the quantum BV operator. We define the \textit{quantum BV operator} $\hat{s}$ as the deformation of\footnote{Note that this is indeed a derivation, even though $\{ \cdot, F \}_\star$ is not a derivation for arbitrary $F$, \cf \cite{RejznerFredenhagenQuantization}.} $\{ \cdot, S_0 \}_\star$ under the action of $R_{S_\ia}$,
\begin{equation}\label{intertwining:s}
\hat{s} \doteq R_{S_\ia}^{-1} \circ \{ \cdot, S_0 \}_\star \circ R_{S_\ia}.
\end{equation}
It was proven in \cite{RejznerFredenhagenQuantization} that the 0-th cohomology of $\hat{s}$ characterizes the gauge invariant quantum observables and it is independent of the choice of gauge fixing fermion $\Psi$, as long as the quantum master equation (QME) 
\begin{equation}\label{suff:condition}
\{ \expT{i S_\ia / \hbar}, S_0 \}_\star = 0
\end{equation}
is fulfilled. This equation can be understood as the gauge invariance of the formal S-matrix. If the QME holds, then (\ref{intertwining:s}) can be rewritten as
\begin{equation}\label{QBV0}
\hat{s} F = \expT{-i S_\ia / \hbar} \T \left( \{ \expT{ i S_\ia / \hbar} \T F, S_0  \}_{\star}\right).
\end{equation}
To make contact with the standard approach \cite{HenneauxTeitelboim}, we note that using the properties of $\star$ and $\T$ one can rewrite  (\ref{suff:condition}) and (\ref{QBV0}) as
\begin{gather*}
\tfrac{1}{2}\{S_0 + S_\ia, S_0 + S_\ia \}_{\sst\TT}=i \hbar \Lap (S_0 + S_\ia), \\
\hat{s} F = \{ F, S_0 + S_\ia \}_{\sst\TT} - i \hbar \Lap F,
\end{gather*}
respectively.

\subsection{Extension to more singular elements}
\label{sec:Extension}
In the next step we want to extend the associative product $\star$ to more singular objects than the elements of $\BV^\nm_\reg(\Sigma)$. The right class of functionals turns out to be $\BV^\nm_{\mc}(M)$. On this space we first introduce the Wick ordering. This can be done in a purely algebraic way, which does not require the existence of a preferred (``vacuum'') state.
For this construction
one needs Hadamard two-point functions, i.e., a set of distributions $\omega^{\alpha \beta}$ on $\Sigma^2$ fulfilling
\begin{align*}
  \omega^{\alpha \beta}(x,y) - (-1)^{\betrag{\phi^\alpha} \betrag{\phi^\beta}} \omega^{\beta \alpha}(y,x)& = i \Pei{\phi^\alpha(x)}{\phi^\beta(y)}, \\
 O^\alpha_\beta \omega^{\beta \gamma} & = 0, \\
 \WF(\omega^{\alpha \beta}) & \subset C_+, \\
 \overline{\omega^{\alpha \beta}(x,y)} & = \omega^{\beta \alpha}(y,x).
\end{align*}
Here $O^\alpha_\beta \phi^\beta = 0$ are the equations of motion, and
\[
 C_+ = \{ (x_1, x_2; k_1, - k_2) \in T^* \Sigma^2 \setminus \{ 0 \} | (x_1; k_1) \sim (x_2; k_2), k_1 \in \bar V^+_{x_1} \},
\]
where $(x_1; k_1) \sim (x_2; k_2)$ if there is a lightlike geodesic from $x_1$ to $x_2$ to which $k_1$ and $k_2$ are coparallel. In the context of gauge theories, one also has to require that
\begin{equation}
\label{eq:gaugeInvariance}
 {\gamma_0}^\alpha_\gamma \omega^{\gamma \beta} + (-1)^\betrag{\phi^\alpha} {\gamma_0}^\beta_\gamma \omega^{\alpha \gamma} = 0,
\end{equation}
where ${\gamma_0}^\alpha_\gamma$ are the coefficient of the free gauge transformation in the basis $\{ \phi^\alpha \}$. This ensures that $\gamma_0$ still acts as a derivation on the deformed algebra. In the next subsection we will prove that Hadamard two-point functions exist on generic on-shell backgrounds. 

Let us pick an arbitrary two-point function $\omega$. 
Defining the Wick polynomials means basically giving a precise mathematical sense to expressions of the form $\WDp{\ph^2(x)}_{\omega} \doteq \lim(\ph(x)\ph(y)-\omega(x,y)\1)$, $x\rightarrow y$. On arbitrary backgrounds one encounters certain difficulty, since the choice of $\omega$ is not unique and cannot be done in a covariant way \cite{HollandsWaldWick}. Nevertheless there is a way out if we make use of the fact that different choices of $\omega$ differ only by a smooth function. Let us write $\omega$ in the form $\omega=\frac{i}{2}\Delta+H$. Following  \cite{BDF09} (see also \cite{Rej11}) we define a transformation
 $\al_{H}: \A_\reg(\Sigma) \rightarrow \A_\reg(\Sigma)$ by
 \[
 \al_{H}\doteq\exp({\hbar\Ga_{H}}),
 \]
where
\[
\Ga_{H}(F) \doteq \sum_{\al,\beta}(-1)^{(1+\betrag{F}) \betrag{\phi^\al}} \int {H}^{\al \beta}(x,y) \frac{\delta^2 F}{\delta\phi^\al(x)\delta\phi^\beta(y)} \mu(x) \mu(y).
\]
This allows us to define another product on $\A_\reg(\Sigma)$, equivalent to $\star$, via
\[
F\star_{H} G\doteq\al_{H}(\al_{H}^{-1}(F)\star \al_{H}^{-1}(G)).
\]
The crucial point is that this product is also well defined on $\BV^\nm_\mc(\Sigma)[[\hbar]]$.

As there is no preferred two-point function, and hence no preferred $H$, we have to consider all of them simultaneously. We denote by $\Had(\Sigma)$ the set of admissible $H$, and define $\A(\Sigma)$ to be the space\footnote{With addition given by $(F+G)_H = F_H + G_H$.} of families $F = \{ F_H \}_{H \in \Had(\Sigma)}$, $F_H \in \BV^\nm_\mc(\Sigma)[[\hbar]]$ fulfilling the relation
\[
 F_{H'} = \exp(\hbar \Gamma_{H'-H}) F_H.
\]
We equip $\A(\Sigma)$ with the product
\[
 (F \star G)_H = F_H \star_H G_H.
\]
The topology is the one inherited from $\BV^\nm_\mc(\Sigma)$ for an arbitrary $H$ (here it is important that $H-H'$ is smooth, so that the topologies are equivalent). The support of $F \in \A(\Sigma)$ is defined as $\supp(F) = \supp(F_H)$. Again, this is independent of $H$. Functional derivatives are defined by
\[
 \skal{\frac{\delta F}{\delta \phi}}{\psi}_H = \skal{\frac{\delta F_H}{\delta \phi}}{\psi},
\]
which is well defined as $\Gamma_{H'-H}$ commutes with functional derivatives.

The assignment $\Sigma \to (\A(\Sigma), \star)$ is a covariant functor between $\CatSub$ and $\CatAlg^*$ which maps a morphism $\chi: \Sigma \to \Sigma'$ to the morphism $\chi_*: \A(\Sigma) \to \A(\Sigma')$ defined by
\[
 (\chi_* F)_H \doteq \chi_*(F_{H|_{\Sigma \times \Sigma}}),
\]
where $\chi_*$ on the \rhs is the morphism of $\BV^\nm_\mc$. One defines $\A_\loc(\Sigma)$ in the analogous way by restricting to $F_H \in \BV^\nm_\loc(\Sigma)[[\hbar]]$. This is a covariant functor between $\CatSub$ and $\CatVec_i$.

One can think of the elements of $\A(\Sigma)$ as Wick powers. To see this, note that, in the case of a scalar field,
\[
 \phi^2(x) \star_H \phi^2(y) = \phi^2(x) \phi^2(y) - 2 \omega(x,y) \phi(x) \phi(y) + 2 \omega(x,y)^2,
\]
which corresponds to Wick's theorem for a two-point function $\omega$.

\subsection{Two-point functions and states}
\label{sec:2pt}

The explicit construction of two-point functions on flat (or very symmetric) backgrounds is rather straightforward. For the proof of the existence of two-point functions on generic backgrounds, one usually proceeds via the deformation argument given in \cite{FNW81}, i.e., one deforms the background in the past of a neighborhood of some Cauchy surface $\Cauchy$ to a flat (or highly symmetric) background, where one can construct two-point functions. These can then be transported, using the equations of motion, to the neighborhood of $\Cauchy$, and from there to the whole original background. In the present setting, we have to be a bit more careful, as after the deformation, we no longer have an on-shell background. Hence, the free BRST operator is not a symmetry on the deformed background, and we have no guarantee that \eqref{eq:gaugeInvariance} is preserved under transport via the equation of motion.

However, we may proceed as follows: We deform the background as indicated above. On this background, we postulate
%\footnote{Note that in the derivation of \eqref{eq:L_0}, we used that the background is on-shell. Hence, at the present point, the choice of the Lagrangean \eqref{eq:L_0_B} is a postulate.}
the Lagrangean \eqref{eq:L_0_B}. 
%\begin{equation}
%\label{eq:L_aux}
% L^\mathrm{aux}(f)(\ph) = \int f \left[ \tfrac{1}{2} (Q \varphi)^a h_{ab} \nabla^\mu \del_\mu (Q \varphi)^b + B^\nu \ud X^a_\nu h_{ab} \nabla^\mu \del^\mu \varphi^b \right] \mu.
%\end{equation}
It coincides with the original Lagrangean where the background is on-shell (note that in the derivation of \eqref{eq:L_0} we used that the background is on-shell). In the region where the deformed background is on-shell and flat (in the sense that $\nabla_\mu \ud X^a_\nu = 0$), one constructs a two-point function corresponding to this action. As we explicitly show below, it can be chosen to be of the form
\begin{equation}
\label{eq:omega}
 \omega = \begin{pmatrix} 0 & \omega_{\mu b} & 0 \\ \omega^{a \nu} & 0 & 0 \\ 0 & 0 & \omega^a_b \end{pmatrix}
\end{equation}
in the case of an open string. Using the equations of motion corresponding to the Lagrangean \eqref{eq:L_0_B}, we can transport it to the neighborhood of $\Cauchy$, where it is then of the form \eqref{eq:Delta}. In order to comply to \eqref{eq:gaugeInvariance} we define the two-point function
\[
 \omega(C^\mu(x) \bar C^\nu(y) ) = - i g^{\mu \lambda} \ud X^a_\lambda(x) h_{ab} \omega(\Phi^b(x) B^\nu(y))
\]
in this neighborhood of $\Cauchy$, and set all the remaining ones to zero. Using the equation of motion, the two-point function can be transported to the whole original background $\Sigma$.

For a generic open string background with boundary conditions as discussed in Remark~\ref{rem:bc}, we note that also the boundary conditions can be deformed (possibly via a linear combination of Dirichlet and Neumann boundary conditions) to the Dirichlet case. Below, we explicitly construct a two-point function for the flat Dirichlet string, whose $(B, \Phi)$ component has the form \eqref{eq:omega}. It follows that by the procedure indicated above, one can construct two-point functions for any on-shell background for the open string. 

We remark that we did not require any positivity, as this was not relevant for the definition of Wick powers in the previous section. However, we will show below that on the flat background one can construct a two-point function such that the physical states (corresponding to $Q\Phi$) are positive definite, whereas those in the image of $\gamma_0$, i.e., those corresponding to $B$ and $C$, are null. By inspection of the form of the fundamental solution \eqref{eq:Delta} and the form of the two-point function \eqref{eq:omega}, we see that this is still the case for the two-point function constructed as above. Hence, the corresponding GNS representation will be positive definite\footnote{\label{ft:posDef}In the sense that $\braket{\Psi}{\Psi} \geq 0$ for $\Psi \in \ker Q_0$ and $\braket{\Psi}{\Psi} = 0$ for $\Psi \in \ker Q_0$ iff $\Psi \in \ran Q_0$.}.

\subsubsection{The Dirichlet string}
We consider the open string with Dirichlet boundary conditions. As a background, we choose the hypersurface given by the coordinates $\tau, \sigma$ as
\[
 \R \times [0,\pi] \ni (\tau, \sigma) \mapsto (\tau, \sigma, 0, \dots, 0) \in M.
\]
The induced metric in these coordinates is $g_{\mu \nu} = \eta_{\mu \nu} = \diag(-1,1)$ and we have $\ud X^a_\mu = \1^a_\mu$.
We consider Dirichlet boundary conditions
\begin{equation*}
 \varphi^a(0) = 0 = \varphi^a(\pi),
\end{equation*}
and analogously for the auxiliary fields.
The linearized equations of motion are
\begin{align*}
 \Box C^\mu & = 0, &
 \Box \bar C_\mu & = 0, \\
 \Box B_\mu & = 0, &
 \Box \Phi^a & = 0.
\end{align*}
According to our discussion in Section~\ref{sec:Peierls}, the non-vanishing free Peierls brackets are given by
\begin{align*}
 \Pei{C^\mu(x)}{\bar C_\nu(y)} & = - i \delta^\mu_\nu \Delta(x,y), \\
 \Pei{B_\mu(x)}{\Phi^a(y)} & = \Delta(x,y) \ud X^a_\mu, \\
 \Pei{\Phi^a(x)}{\Phi^b(y)} & = \Delta(x,y) k^{ab}.
\end{align*}
Here $\Delta$ is the scalar causal propagator corresponding to Dirichlet boundary conditions and $k^{ab} = \diag(0,0,1, \dots, 1)$.

In order to construct corresponding states, we proceed as in \cite{HollandsYM}. For $n \in \N$ we consider the orthonormal positive frequency solutions
\[
 u_n(\tau, \sigma) = \sqrt{\tfrac{2}{\pi}} \sin(n \sigma) e^{i n \tau}
\]
and set
\[
 \omega(x, y) = \sum_{n \in \N} d_n u_n(x)^* u_n(y),  
\]
with $d_n = \frac{1}{2n}$.
We may now set
\begin{align*}
 \omega_2(C^\mu(x) \bar C_\nu(y)) & = - i \delta^\mu_\nu \omega(x,y), \\
 \omega_2(\Phi^a(x) \Phi^b(y)) & = k^{ab} \omega(x,y), \\
 \omega_2(B_\mu(x) \Phi^a(y)) & = \ud X^a_\mu \omega(x,y),
\end{align*}
with all other combinations (also involving antifields) vanishing.
A Hilbert space representation compatible with the above two-point function may be constructed as follows: Consider the symmetric Fock space $\HS_\varphi$ corresponding to the one-particle space generated by elements $e^a_{m}$, $f_{\mu m}$, $m \in \N$, $a \in \{0, \dots, n-1 \}$ with indefinite inner product
\begin{align*}
 (e^a_{n}, e^b_{m}) & = k^{ab} \delta_{mn}, & (f_{\mu m}, f_{\nu n}) & = 0, & (e^a_{m}, f_{\nu n}) & = \delta_{mn} \ud X^a_\nu.
\end{align*}
Also consider the antisymmetric Fock space $\HS_c$ corresponding to the one-particle space generated by elements $g_{\mu s n}$, with $s \in \{1, 2\}$ and inner product $(g_{\mu s n}, g_{\nu t m}) = i g_{\mu \nu} \eps_{st} \delta_{m n}$, where $\eps$ is the standard symplectic $2 \times 2$ matrix. Also consider the standard creation operators ${a^a_m}^+, b^+_{\mu m}, c^+_{\mu s m}$ that create $e^a_{m}$, $f_{\mu m}$ and $g_{\mu s m}$, and corresponding annihilation operators $a^a_{m}, b_{\mu m}, c_{\mu s m}$ that are normalized as
\begin{align*}
 [a^a_m, {a^b_n}^+] & = h^{ab} \delta_{mn}, & [b_{\mu m}, b_{\nu n}^+] & = \eta_{\mu \nu} \delta_{mn}, & \{ c_{\mu s m}, c^+_{\nu t n} \} & = \eta_{\mu \nu} \delta_{m n} \delta_{st}.
\end{align*}
On $\HS_\varphi \otimes \HS_c$  consider the representation
\begin{align*}
 \pi(\Phi^a(x)) & =
\begin{cases}
\sum_{n \in \N} \sqrt{d_n} \left( u_n(x) {a^a_n}^+ + u_n(x)^* a^a_{n} \right) & a \geq 2, \\
\sum_{n \in \N} \sqrt{d_n} \left( u_n(x) {a^a_n}^+ + u_n(x)^* \ud X^a_\mu \eta^{\mu \nu} b_{\nu n} \right) & a \in \{ 0, 1 \}, 
\end{cases} \\
 \pi(C_\mu(x)) & = \sum_{n \in \N} \sqrt{d_n}  \left( u_n(x) c^+_{\mu 2 n} - i u_n(x)^* c_{\mu 1 n} \right), \\
 \pi(\bar C_\mu(x)) & = \sum_{n \in \N} \sqrt{d_n}  \left( u_n(x) c^+_{\mu 1 n} + i u_n(x)^* c_{\mu 2 n} \right), \\
 \pi(B_\mu(x)) & = \sum_{n \in \N} \sqrt{d_n} \left( u_n(x) b^+_{\mu n} + u_n(x)^* \ud X^a_\mu a^b_{n} h_{ab} \right).
\end{align*}
Using \eqref{eq:current} and
\[
 \int_0^\pi u_n(0,\sigma)^{(*)} u_m(0,\sigma) \ud \sigma = \delta_{m n},
\]
we compute the free BRST charge
\[
  \pi(Q_0) = - \eta^{\mu \nu} \sum_{n \in \N} \left( b^+_{\nu n} c_{\mu 1 n} - i \ud X^a_\nu h_{ab} a^b_{n} c^+_{\mu 2 n} \right),
\]
which is indeed nilpotent. We see that the transversal fluctuations, i.e., those corresponding to $\Phi^a$ for $a \geq 2$ are in the kernel of $\pi(Q_0)$, whereas those corresponding to $C^\mu$ and $B_\nu$ are even in the image of $\pi(Q_0)$. Hence, we have a positive definite representation, in the sense defined in footnote~\ref{ft:posDef}.

\subsection{Renormalization}
\label{renorm}
We already extended the operator product $\star$ to more singular objects, i.e., elements of $\A(\Sigma)$. We now want to do the same also for the time ordered product $\T$. Here the situation is more complicated. In fact, in order to control the ambiguities present in this extension, it is advantageous to proceed differently than in Section~\ref{sec:Extension}.

We recall that multilocal functionals were defined as the space of finite products of local ones, and microcausal ones as a certain completion thereof. The first goal will thus be to define time-ordered products as linear maps
\[
 \TTR^{\,k}: \BV^\nm_\loc(\Sigma)^{\otimes k} \to \A(\Sigma).
\]
They should be natural transformations between the functors ${\BV^\nm_\loc}^{\otimes k}$ and $\A$. One requires the following properties:
\begin{description}
\item[Starting element.] For the lowest order time-ordered products we require $\TTR^{0}=0$, $(\TTR^{1} F)_H = \exp(\hbar \Gamma_{H-h}) F$. Here $h$ is the symmetric part of the Hadamard parametrix.\footnote{The usage of $h$ ensures the covariance of the construction. This corresponds to the locally covariant Wick powers introduced in \cite{HollandsWaldWick}. Note that $h$ is well-defined only locally, but the coinciding point limits of all its derivatives are unique, so that the expression is well-defined on local functionals.}
\item[Supports.]  $\supp\TTR^{k}(F_1,\dots,F_k)\subset\bigcup\limits_{i=1}^k\supp F_i$.
\item[Symmetry.] The time ordered products are graded-symmetric under a permutation of factors.
\item[Unitarity.]
Let $\overline{\TTR}^{k}(\otimes_i F_i) =
[\TTR^{k}(\otimes_i F_i^*)]^*$ be
the antitime-ordered product. Then we require
\[
\overline{\TTR}^{k} \bigg( \bigotimes_{i=1}^k F_i \bigg) =
\sum_{I_1 \sqcup \dots \sqcup I_j = \underline{k}}
(-1)^{k + j} \TTR^{\betrag{I_1}}\bigg(
\bigotimes_{i \in I_1} F_i \bigg) \star \dots \star
\TTR^{\betrag{I_j}}\bigg(\bigotimes_{j \in I_j} F_j \bigg),
\]
where the sum runs over all partitions of the set $ \underline{k}\doteq\{1, \dots, k\}$ into
pairwise disjoint subsets $I_1, \dots, I_j$.
\item[Causal Factorization.]
If $J^+(\supp F_l) \cap \supp F_m = \emptyset$ for all $1 \leq l \leq i$, $i+1 \leq m \leq k$, then
\begin{equation*}
\TTR^{k}(F_1\otimes \dots \otimes F_k)=
\TTR^{i}(F_1\otimes \dots \otimes F_i) \star
\TTR^{k-i}(F_{i+1} \otimes \dots \otimes F_k).
\end{equation*}
\item[Field equation.]
The free field equation is implemented in a Schwinger--Dyson type equation:
\begin{multline}\label{fieldeq}
\TTR^{\,k+1}\bigg(
\frac{\delta S_0}{\delta \phi^\al(x)} \otimes \bigotimes_{i=1}^k F_i
\bigg)
= i \hbar \sum_i
\TTR^{\,k}
\bigg(
F_1 \otimes \cdots \frac{\delta F_i}{\delta
    \phi^\al(x)} \otimes \cdots F_k)
\bigg)+\\+\frac{\delta S_0}{\delta \phi^\al(x)} \star \TTR^{\,k}
\bigg(
F_1 \otimes \cdots  F_k
\bigg)
\end{multline}
\end{description}
Further requirements are field independence and $\ph$-locality, \cf \cite{BDF09}, and the microlocal spectrum condition, scaling\footnote{In the present context, the scaling condition is a condition on the relation between time-ordered products on $\Sigma$ and $\eta \Sigma$, where $\eta$ is a positive scale parameter.}, and the smooth (analytic) dependence on the background fields ($\ud X$ and $g$), \cf \cite{HollandsWaldTO}. Renormalized time-ordered products fulfilling these properties were constructed in \cite{HollandsWaldTO} for the scalar field and in \cite{HollandsYM} for Yang--Mills fields. The techniques are straightforwardly generalizable to the present case. This is sketched in Section~\ref{sec:Poincare}, where we show that also covariance under global Poincar\'e transformations can be kept.

With the above assumptions, $\TTR^k$ is uniquely fixed by the lower order maps $\TTR^m$, $m<k$, up to the addition of a $k$-linear map
\[
Z_k: \BV^\nm_\loc(\Sigma)^{\otimes k} \to \A_\loc(\Sigma),
\]
which describes the possible finite renormalizations (and is again a natural transformation). This renormalization freedom is characterized by the renormalization group in the sense of St\"uckelberg--Petermann \cite{DuetschFredenhagenAWI}. Its relation to different notions of the renormalization group such as the Gell-Mann--Low or the Wilson renormalization group is discussed in \cite{BDF09}.

Having defined the renormalized $k$-fold time ordered products $\TTR^k$, one can proceed to define the renormalized time ordered product $\TR$ as a binary operation on an appropriate domain in $\A(\Sigma)$. Let $\dot \BV^{\nm}_\loc(\Sigma)$ be the space of local functionals which vanish at $\phi^\al=(0,0,0,0)$, and let $S \dot \BV^\nm_\loc(\Sigma)$ denote the corresponding space of graded symmetric tensor powers. It was shown in \cite{RejznerFredenhagenQuantization} that the pointwise multiplication $m: S \dot \BV^\nm_\loc(\Sigma) \rightarrow \BV^\nm(\Sigma)$ is bijective. Therefore we can use its inverse $m^{-1}$ to map multilocal functionals to $S \dot \BV^\nm_\loc(\Sigma)$, where the maps $\TTR^k$ have their domains.
The renormalized time ordering operator $\TTR$ can then be defined as a natural transformation between $\BV^\nm$ and $\A$ given by
\begin{equation*}
\TTR\doteq(\bigoplus_k  \TTR^k)\circ m^{-1}.
\end{equation*}
This operator is a formal power series in $\hbar$ starting with the identity, hence it is invertible on its image, and the renormalized time ordered product is now defined on the image $\TTR(\BV^\nm(\Sigma)[[\hbar]])$ of $\TTR$ by 
\begin{equation*}
F\TRH G\doteq\TTR(\TTR^{-1}F\cdot\TTR^{-1}G).
\end{equation*}
The renormalized formal S-matrix is defined as in (\ref{Smatrix}), but with $\T$ replaced by $\TR$.

An important result in causal perturbation theory is the main theorem of renormalization \cite{BDF09}. It states that
given two renormalized S-matrices $\Scal$ and $\widehat{\Scal}$ 
satisfying the conditions Causality, Starting Element, $\varphi$-locality, and 
Field Independence (see  \cite{BDF09} for details), there exists a unique element of the St\"uckelberg--Petermann renormalization group $Z\in\mathcal{R}$ such that 
\begin{equation}
\widehat{\Scal}=\Scal\circ Z.\label{mainthm}
\end{equation}
Conversely, given an S-matrix $\Scal$ satisfying the 
mentioned conditions and a $Z\in\mathcal{R}$, equation (\ref{mainthm})
defines a new S-matrix $\widehat{\Scal}$ also satisfying these conditions. 

An important consequence of the fact that $\TR$ can be defined as a binary associative operation is the possibility to obtain the renormalized QME and the renormalized BV operator in a natural way. This was done in \cite{RejznerFredenhagenQuantization}. We review here shorty the results obtained there. The renormalized QME and the renormalized BV operator can be algebraically written as
\begin{align*}
0 & = \{ \expTR{i S_\ia^\tau / \hbar}, S^\tau_0 \}_{\star},\\
\hat{s}(F) & = \expTR{- i S_\ia^\tau / \hbar} \TR \{ \expTR{i S_\ia^\tau / \hbar} \TR F, S^\tau_0 \}_{\star},
\end{align*}
where $F \in \TTR(\BV^\nm(\Sigma)[[\hbar]])$ and $S^\tau_0 = \TTR^1 S_0$, $S_\ia^\tau = \TTR^1 S_\ia$. This has to be understood in the sense of formal power series in $\hbar$, $\lambda$, and the antifields.
These formulas can be simplified using the Master Ward Identity \cite{BreDue,HollandsYM}. Then they take the form
\begin{align}
0 & = \tfrac{1}{2} \{ S^\tau_0 + S_\ia^\tau, S^\tau_0 + S_\ia^\tau \}_{\TTR} - \Lap_{S_\ia^\tau}(S_\ia^\tau), \nonumber \\
\label{eq:hat_s}
\hat{s} F & = \{F, S^\tau_0 + S_\ia^\tau \}_{\TTR} - \Lap_{S_\ia^\tau}(F),
\end{align}
where $\Lap_{S_\ia^\tau}$ a linear map
$\Lap_{S_\ia^\tau}:\TTR(\BV^\nm_\loc(\Sigma)[[\hbar]]) \to \A_\loc(\Sigma)[[\lambda]]$. We can think of it as the renormalized version of the graded Laplacian $\Lap$.

\subsection{Poincar\'e covariance}
\label{sec:Poincare}
For Minkowski space as target space, the action has another, global, symmetry, namely the Poincar\'e transformations of the target space. We want to ensure that our model respects this symmetry.
However, the choice of a background $\Sigma$ breaks this symmetry. Hence, we can not expect to be able to represent the Poincar\'e group on $\BV^\nm(\Sigma)$ for some fixed $\Sigma$.\footnote{For Dirichlet boundary conditions, Lorentz transformations are not even a symmetry on a fixed background. Any implementation of Lorentz transformations on a fixed background has the unpleasant feature of spoiling the perturbative expansion, as $\delta_\Lambda$ is of $\order(\lambda^{-1})$ is general.}
But we can ensure covariance if we also transform the background $\Sigma$. This not only guarantees the existence of a representation of the stabilizer group $\Poin_\Sigma$ of $\Sigma$ on $\BV^\nm(\Sigma)$, \cf the end of this subsection, but is also crucial in order to retain a Lorentz invariant Lagrangean in the renormalized setting, \cf the discussion of possible counterterms in Section~\ref{sec:Cohomology}.

The action $\alpha$ of the proper orthochronous Poincar\'e group $\Poin$ on $M$ transforms the background $\Sigma$, i.e., we have an action of $\Poin$ on $\CatSub$. For an element $g \in \Poin$, we can define a functor $\E_g$ by $\E_g(\Sigma) = \E(\alpha_g \Sigma)$. Each $g \in \Poin$ thus induces a natural transformation $\beta_g: \E \to \E_g$. Concretely, we have $(\beta_g \varphi)^a(x) = \Lambda^a_b \varphi^b(\alpha_g^{-1} x)$, where $\Lambda$ is the Lorentz transformation part of $g$, and $\alpha_g$ is now the map $\Sigma \to \alpha_g \Sigma$. It follows that indeed $\alpha_g ( \tilde X[\varphi](x)) = \tilde X[\beta_g \varphi](\alpha_g x)$, i.e., the image of the Poincar\'e transformed $\varphi$ coincides with the Poincar\'e transformed image of $\varphi$. On the functor $\F$, one acts via pullback, i.e., for $F \in \F(\Sigma)$ define $(\beta_g F)(\varphi) = F(\beta^{-1}_g \varphi)$ for $\varphi \in \E(\alpha_g \Sigma)$. Obviously, this action preserves the algebra structure of $\F$, i.e., for each $g \in \Poin$ we have a natural transformation $\beta_g: \F \to \F_g$, even when $\F$ is considered as a functor from $\CatSub$ to $\CatAlg$. All this can be straightforwardly extended to the $\BV^\nm$ functor. As $\alpha$ is a global symmetry of the action, it follows that the retarded and advanced propagator transform covariantly under $\alpha$. Hence, the same it true for the fundamental solution and thus for the canonical structure. 

It remains to ensure covariance at the quantum level. For this, it is advantageous to take another point of view on the model. For fixed $\Sigma$, we can view it as a quantum field theory on $\Sigma$, where the fields are sections of $T\Sigma$ and $TM$, the latter being some trivial vector bundle. Furthermore, there is a background field $\ud X$, which is a section of $T^*\Sigma \otimes TM$, and a nondegenerate bilinear form $h$ on $TM$. The metric $g$ happens to be given by $g = h(\ud X, \ud X)$. There is a global action $\alpha$ of the proper orthochronous Lorentz group $SO_0(1,n-1)$ on $TM$ that leaves $h$ invariant.\footnote{In this picture, where the target space is no longer present, Poincar\'e transformations act trivially on $\Sigma$ and $T \Sigma$. In the original picture, this corresponds to identifying the image of $\Sigma$ under Poincar\'e transformation with $\Sigma$ itself, by an isometry.} As this is also an invariance of the action, the Feynman parametrix respects this symmetry. Hence the only place where a noncovariant behavior can occur is renormalization, i.e., the extension of ill-defined products of distributions (time-ordered products) to the diagonal $\Delta_k \in \Sigma^k$. As shown in \cite{HollandsWaldTO}, this extension can be done in a local and covariant way, i.e., such that it only depends on local data and is independent of the choice of a coordinate system on $\Sigma$. A crucial point here is that one can ensure local Lorentz invariance, i.e., the Lorentz invariance of the extension to $\R^{d(k-1)}$ of a Lorentz invariant tensor-valued distribution on $\R^{d(k-1)} \setminus {0}$, identified with the Riemannian normal coordinates in $\{x\} \times \Sigma^{k-1}$ around $(x, \dots, x)$. The argument given in \cite{HollandsWaldTO} relies on the triviality of the first cohomology class of $SO_0(1, d-1)$ for finite-dimensional representations, which holds in the case $d >2$.\footnote{For the Lie algebra cohomology, this follows from Whitehead's lemma. To pass to the continuous group cohomology, one may use the van Est theorem \cite[Cor.~III.7.2]{Guichardet}.}
For $d=2$, this triviality does not hold, but we have $SO_0(1,1) \simeq \R$, so there is a single generator of the Lie algebra. Its kernel in a given representation characterizes the invariant vectors. Hence, one may use the procedure presented in \cite[App.~D]{DuetschFredenhagenAWI} to obtain a Lorentz invariant extension (one simply replaces the Casimir operator $C_0$ used there by the generator). Hence, we can ensure local Lorentz covariance, and analogously, covariance \wrt target space Lorentz transformations. In this way we have preserved the functoriality (by locality and covariance), translation covariance (as the construction was independent of the actual position of $\Sigma$ in $M$), and global Lorentz covariance.

Let us now consider the case where a subgroup $\Poin_\Sigma$ of the Poincar\'e group leaves $\Sigma$ invariant. We then have a continuous action of $\Poin_\Sigma$ on $\BV^\nm(\Sigma)$. Considering a state $\omega$ that is invariant under this action, we obtain a continuous representation of $\Poin_\Sigma$ on the GNS Hilbert space $\HS_\omega$ such that $\pi(g) \Omega_\omega = 0$ for all $g \in \Poin_\Sigma$.

\subsection{QME in the algebraic adiabatic limit}
\label{sec:QME}

In causal perturbation theory one has to work with interactions that are localized, but in physical situations there is usually no natural cutoff.
This problem can be solved, on the algebraic level, by the construction of the so-called algebraic adiabatic limit \cite{BrunettiFredenhagenScalingDegree}: In order to define the interacting observables localized in a region $\mathcal{O}$, one cuts off the interaction with a test function that is identically 1 on a neighborhood of $\mathcal{O}$. One can show that the algebras obtained by different choices of the cut-off function are unitarily equivalent. This construction naturally fits into 
the formulation of the classical and quantum theory, where the Lagrangeans are treated as natural transformations \cite{BDF09}. This point of view was also taken in \cite{RejznerFredenhagen} to formulate the classical master equation and in \cite{RejznerFredenhagenQuantization} for the quantum master equation. Here we give only a short review of this formulation.

Let $L_0$ be the free generalized Lagrangean at vanishing antifield number, and $L_\ia$ the interaction term. Both are now to be understood as natural transformations between $\D$ and $\BV^\nm$. The classical master equation (CME) is formulated as the condition that
\begin{equation*}
\{ L_0 + L_\ia, L_0 + L_\ia \} \sim 0,
\end{equation*}
with the equivalence relation defined in (\ref{equivalence}). Given $L_0$ and $L_\ia$ fulfilling this condition we construct natural transformations $\TTR^1({L_0})$ and   $\TTR^1({L_\ia})$  from $\D$ to $\A$. We denote the corresponding equivalence classes by $S_0$ and $S_\ia$ and it holds $\{S_0 + S_\ia, S_0 + S_\ia \}_{\TTR} \sim 0$. The quantum master equation is a statement that the S-matrix in the algebraic adiabatic limit is invariant under the quantum BV operator, i.e.,
\[
\supp \left( \expTR{-i{L_\ia^\tau}_\Sigma(f_1) / \hbar} \TR \left( \{ \expTR{i{L_\ia^\tau}_\Sigma(f_1) / \hbar},{L_0^\tau}_\Sigma(f)\}_{\star} \right) \right) \subset \supp \ud f \cup \supp \ud f_1.
\]
Using  the  Master Ward Identity, one finds that this expression is again an element of $\A_\loc(\Sigma)$, so the condition above can be also formulated on the level of natural transformations,
\begin{equation*}
 \expTR{-i S_\ia / \hbar} \TR \left( \{ \expTR{i S_\ia / \hbar}, S_0 \}_{\star} \right) \sim 0,
\end{equation*}
which can be written more explicitly as
\begin{equation}\label{eQMEr1}
 \tfrac{1}{2} \{ S_0 + S_\ia, S_0 + S_\ia \}_{\TTR} - \Lap_{S_\ia}(S_\ia) \sim 0.
\end{equation}
The central question in the BV quantization is the possibility to fulfill the QME in the above form. If $L_0$, $L_\ia$ satisfy the CME but $\Lap_{S_\ia}(S_\ia) \neq 0$ we can use the standard methods of homological perturbation theory and postulate a solution of the QME of the form $W = \sum_n \hbar^n W_n$, $W_0=S_\ia$. Following \cite{HenneauxTeitelboim}, we want to construct $W_n$ inductively in each order. From the nilpotency of $\hat s$ it follows that we obtain some further consistency conditions on $\Lap_{S_\ia}$, in particular $\{\Lap^1_{S_\ia}(S_\ia), S_0 + S_\ia \}_{\TTR} \sim 0$. Using these additional constraints one can show that the QME has a solution if the cohomology of $s$ on the space of actions vanishes in ghost number $\#\gh=1$. This effectively amounts to calculate the cohomology of $s$ modulo $\ud$ on the space of local forms of ghost number $\#\gh=1$. In Section~\ref{sec:Cohomology}, we
prove that for the Nambu--Goto action this cohomology does not contain any nontrivial covariant elements. This shows that the QME can be fulfilled and the string, and its higher-dimensional generalizations, can be perturbatively quantized.

\subsection{Further quantization conditions}
Apart from the QME, there are two further quantization conditions one should impose. In order to ensure that an observable $F$, i.e., an $s$ invariant functional of ghost number 0, is still an observable after quantization, one needs
\begin{equation}
\label{eq:12c}
 \hat s F = 0. 
\end{equation}
Here $\hat s$ is defined by \eqref{eq:hat_s} with the cutoff function in $S_\ia$ chosen to be $1$ on $\supp F$. A detailed discussion of this condition in the context of Yang--Mills theories can be found in \cite{HollandsYM}. The obstructions to this equation are governed by the cohomology $H(s)$ at ghost number 1.
%\begin{equation}
%\label{eq:12c}
% \tilde s F = 0. 
%\end{equation}
%Here $\tilde s$ is the sum of $\hat s$, defined by \eqref{eq:hat_s} with the cutoff function in $S_\ia$ chosen to be $1$ on $\supp F$, and the correction term present in \eqref{eq:gammaFields}, i.e., the action of $\rho(C)$ on the test tensor section. A detailed discussion of this condition in the context of Yang--Mills theories can be found in \cite{HollandsYM}. The obstructions to this equation are governed by the cohomology $H(s)$ at ghost number 1.

In order to achieve the the nilpotency of the interacting BRST charge, one has to ensure that
\begin{equation}
\label{eq:12b}
 \tilde s \int t_\mu J^\mu \mu = 0,
\end{equation}
where $J^\mu$ is the BRST current and $t_\mu$ is closed, i.e., $\nabla_\mu t_\nu - \nabla_\nu t_\mu = 0$. Again, we refer to \cite{HollandsYM} for details. The obstructions are governed by the cohomology $H(s|\ud)$ at ghost number 2 and form degree $d-1$.

\section{The cohomological structure}
\label{sec:Cohomology}

It is well-known that there is a bijection between the cohomology of $s$ and a subspace of the cohomology of the free part $s_0$ (the so-called Abelian cohomology). This follows from the following theorem, \cf \cite[Prop.~5.6]{PiguetSorella}, for example. For the convenience of the reader, we included a proof in the appendix.

\begin{theorem}
\label{thm:filtration}
Let $N$ be an operator on the space $\Omega$ of local forms with eigenvalues in the nonnegative integers, which commutes with $\ud$. Assume that a corresponding filtration of $\Omega$ and $s$ exists, i.e.,
\begin{align*}
 \Omega & = \bigoplus_{i \in \N_0} \Omega_i, & N \Omega_i & = i \Omega_i, & s_i \Omega_j & \subset \Omega_{i+j}.
\end{align*}
Also assume that $s_0$ is the lowest nontrivial component in the filtrations of $s$.
Then $s_0$ is nilpotent, and there are linear injective maps $\pi: H(s) \to H(s_0)$ and $\pi': H(s | \ud) \to H(s_0 | \ud)$.
\end{theorem}

It follows that if we can show that the relevant cohomologies of $s_0$ are trivial, we proved the triviality of the corresponding cohomologies of $s$. We note that such a perturbative approach is natural, as we are also performing the renormalization order by order, and will thus try to lift violations of Ward identities in the same way. For our purposes, the following cohomology classes are relevant:

%As usual \cite{BarnichBrandtHenneaux00}, we are working on jet space, i.e., we are considering local forms. 
%For our purposes, the following cohomology classes are relevant:
\begin{enumerate}
\item The cohomology of fields, which governs the possibility to renormalize a gauge invariant field so that the result is invariant under the renormalized gauge transformation, \cf \eqref{eq:12c}. 
This corresponds to $H(s)$ at form degree $d$ and ghost number 1.
%In standard gauge theories, this would correspond to $H(s)$ at form degree $d$ and ghost number 1. In the present case, we have to consider that also the test section transforms, \cf \eqref{eq:gammaFields}, so we have look for tensorial $d$-forms that transform under $s$ as under the Lie derivative. However, as discussed in Section~\ref{sec:Fields}, the action of $s$ on the test section is of $\order(\lambda)$, so it does not contribute to $s_0$. Hence, working perturbatively, it suffices to consider $H(s_0)$.
\item The cohomology of $s$ in the class of Lagrangeans, i.e., fields of the form $L(f)$ with $f \in C^\infty_c(\Sigma)$, modulo the equivalence relation \eqref{equivalence}. In standard terminology, this corresponds to $H(s|\ud)$ at form degree $d$. We furthermore require the $d$-form to be a scalar under target space Lorentz transformations and coordinate changes on $\Sigma$. At ghost number 0, this classifies the renormalization ambiguities (counterterms), and at ghost number 1 it gives the potential gauge anomalies.
%The potential obstructions to achieve independence of the observables from the gauge fixing are given by $H^2(s | \ud)$, \cf the discussion in Section~\ref{sec:BackgroundIndependence}.

\item The cohomology of $s$ at ghost number 2 in the space of fields of the form $Q(f)$, with $f \in \Gamma^\infty_c(\Sigma, T^* \Sigma)$, $\nabla_\nu f_\mu - \nabla_\mu f_\nu = 0$, \cf \eqref{eq:12b}. 
In standard terminology, this corresponds to $H(s| \ud)$ at ghost number 2 and form degree $d-1$. 
This cohomology governs the obstructions to achieve a nilpotent interacting BRST charge.
\end{enumerate}

Hence, due to the following proposition, there are no anomalies of any of the above type:

\begin{proposition}
The cohomologies $H(s_0)$ and $H(s_0|\ud)$ vanish at positive ghost number.
\end{proposition} 
\begin{proof}
For the the action of $s_0$ on the components $P \Phi$ and $Q \Phi$ and the corresponding antifields, we have
\begin{align*}
 s_0 (Q \Phi)^a & = 0, & s_0 (P \Phi)^a & = \ud X^a_\mu C^\mu, \\
 s_0 (Q \Phi^\ddagger)_a & =  h_{ab} \nabla^\mu \del_\mu (Q 
 \Phi)^b, & s_0 (P \Phi^\ddagger)_a & = 0.
\end{align*}
As $\ud X$ has maximal rank, it follows that the $P \Phi$'s and $C$'s and their derivatives form contractible pairs \cite{BarnichBrandtHenneaux00}, so they can be removed from the cohomology of $s_0$. Hence, the cohomology of $s_0$ does not contain any ghosts, and is thus trivial at positive ghost number.

To show this for $H(s_0 | \ud)$, we can proceed as follows: We define an operator
\[
 r \doteq \sum_{\underline \mu} \del_{\underline \mu} \left( h_{ab} g^{\nu \lambda} \ud X_\lambda^b \Phi^a \right) \del_{\del_{\underline \mu} C^\nu},
% r \doteq \sum_{\underline \mu} \del_{\underline \mu} \left( {\ud X^{-1}}^\nu_a \Phi^a \right) \del_{\del_{\underline \mu} C^\nu},
\]
%where $\ud X^{-1}$ fulfills ${\ud X^{-1}}^\mu_a \ud X^a_\nu = \delta^\mu_\nu$ and $\ud X^a_\mu {\ud X^{-1}}^\mu_b = P^a_b$ and we sum over multiindices $\underline \mu$.
where we sum over multiindices $\underline \mu$.
Then the operator $N \doteq \{ s_0, r \}$ counts the number of elements of the trivial pair. This operator anticommutes with $\ud$, so by the basic lemma of \cite{BrandtDragonKreuzer89}, also the cohomology $H(s_0 | \ud)$ is trivial at positive ghost number. \qed
\end{proof}

\begin{remark}
Let us compare with the situation for the Polyakov action. There, one also has an auxiliary dynamical metric $g_{\mu \nu}$ and a scalar ghost $c$, and the corresponding antifields. The extended action reads
\begin{multline*}
 S = \int \Big[ \del_\mu \tilde X^a \del_\nu \tilde X^b h_{ab} g^{\mu \nu} + \tilde X_a^\ddagger c^\mu \del_\mu \tilde X^a + c_\mu^\ddagger c^\lambda \del_\lambda c^\mu + c^\ddagger c^\lambda \del_\lambda c \\
 + g_{\mu \nu}^\ddagger \left( c^\lambda \del_\lambda g^{\mu \nu} - \del_\lambda c^\mu g^{\lambda \nu} - \del_\lambda c^\nu g^{\lambda \mu} - c g^{\mu \nu} \right) \Big] \sqrt{-g} \ud x.
\end{multline*}
It is crucial to note that, contrary to our setting, the $\tilde X$ have no component of $\order(\lambda^0)$, i.e., one expands around an event, and not a 2-dimensional object. It follows that the ghosts may not be removed from the cohomology (for discussions of the cohomology for the above action, we refer to \cite{BrandtTroostVanProeyen96,GomisParisSamuel}).
%they are matter fields in the terminology of \cite{BrandtDragonKreuzer90}\footnote{These authors define matter fields as fields of vanishing ghost number that vanish under the free gauge transformation $\gamma_0$.}. Hence, the ghosts may not be removed from the cohomology (for discussions of the cohomology for the above action, we refer to \cite{BrandtTroostVanProeyen96,GomisParisSamuel}). We note that this is not just semantics: Equipping $\tilde X$ with a component of $\order(\lambda^0)$ also changes the physical field content, as the longitudinal fluctuations (\wrt the background) would become unphysical, as in our setting.
\end{remark}

To close the discussion, we briefly comment on the form of the admissible counterterms. To begin, we note that they may depend on the undifferentiated fields $\Phi^a$ (to be more precise, their normal components). This is not in contradiction to translation invariance, as these are implemented solely by translating the background. For the discussion of the higher derivative terms, let us, for simplicity, restrict to the case of a flat background, i.e., $\nabla \ud X = 0$. The elements in the kernel of $\gamma_0$ are then $Q^a_b \nabla_{\underline \mu} \Phi^b$ with $Q$ the projector on the normal bundle introduced in Section~\ref{sec:Peierls}. In particular, this includes the linearized second fundamental form $\tilde K_{\mu \nu}^a = \tilde \nabla_\mu \ud \tilde X^a_\nu$. Hence, we are dealing with a modified gravity theory that also takes the embedding into an ambient space into account. The appearance of extrinsic curvature terms in the context of string theory was discussed already in \cite{Polyakov86,Kleinert}.

\subsection*{Acknowledgements}
%\begin{acknowledgement}
This work was supported by the German Research Foundation (Deutsche Forschungsgemeinschaft (DFG)) through the Institutional Strategy of the University of G\"ot-tingen. %K.R.\ would like to thank Klaus Fredenhagen for enlightening discussions and comments. D.B.\ and J.Z.\ would like to thank Hermann Nicolai and Karl--Henning Rehren for helpful discussions.
We would like to thank Klaus Fredenhagen, Hendrik Grundling, Stefan Hollands, Hermann Nicolai, and Karl--Henning Rehren for helpful discussions and comments.
%\end{acknowledgement}

\appendix

\section*{Appendix}
%\begin{theopargself}
\begin{proof}[Proof of Proposition \ref{prop:dX}]
We set $\lambda = 1$. Let $\phi_\xi: \R \to \Sigma$ be the unique geodesic on $\Sigma$ with $\phi_\xi(0) = x$ and $\del_\tau \phi_\xi(0) = \xi$. Then the map
\[
 [0,1] \times (- \varepsilon, \varepsilon) \ni (s, t) \mapsto \sigma(s, t) = \exp_{X(\phi_\xi(t))}(s \varphi(\phi_\xi(t)))
\]
is a variation through geodesics \cite[Section IX.2]{LangDiffGeo} of the geodesic $\alpha$. Thus, by \cite[Prop. IX.2.8]{LangDiffGeo}, the vector field along $\alpha$,
\[
 \eta(s) = \del_t \sigma(s,t)|_{t=0}
\]
is a Jacobi field. The initial conditions follow from
\begin{gather*}
 \eta(0) = \del_t \sigma(0,t)|_{t=0} = \xi(x), \\
 \nabla_s \eta(0) = \nabla_s \del_t \sigma(s,t)|_{s = t = 0} = \nabla_t \del_s \sigma(s,t)|_{s = t = 0} = \nabla_t \varphi(\phi_\xi(t))|_{t=0} = \nabla_\xi \varphi(x).
\end{gather*}
In the second line we used \cite[Lemma VIII.5.3]{LangDiffGeo} to commute the derivatives \wrt $s$ and $t$.
%\qed
\end{proof}
%\end{theopargself}

%\begin{theopargself}
\begin{proof}[Proof of Theorem~\ref{thm:filtration}]
%\begin{proof}[Proof of Theorem~\ref{thm:filtration}]
Nilpotency of $s_0$ follows from the nilpotency of $s$. Let $\omega$ be a representative of a nonvanishing cohomology class $[\omega] \in H(s | \ud)$. Consider its lowest nonvanishing component $\omega_i \in \Omega_i$. It follows that
\begin{equation*}
 s_0 \omega_i + \ud \nu_i = 0.
\end{equation*}
Assume that $\omega_i$ is trivial in $H(s_0 | \ud)$, i.e.,
\[
 \omega_i = s_0 \hat \omega_i + \ud \hat \nu_i.
\]
Define
\[
 \omega' = \omega - s \hat \omega_i - \ud \hat \nu_i.
\]
This is also an element of $[\omega]$, whose lowest nonvanishing component $\omega'_{i'}$ is at a higher level $i' > i$. At some $i$, this procedure has to stop, as otherwise $[ \omega ]$ would have a representative whose components all vanish, i.e., $[ \omega ] = 0$. At the point where this recursion stops, we define $\pi'([\omega]) = [\omega_i]$, where on the \rhs we take the equivalence class in $H(s_0 | \ud)$. It remains to show that this map is unique. So assume that, by taking another representative $\tilde \omega \in [\omega]$, we arrive at a different equivalence class $[\tilde \omega_{\tilde \imath}] \in H(s_0 | \ud)$. Assume $i < \tilde \imath$. As $\omega - \tilde \omega \sim 0$, we must have
\[
 s_0 \omega_i + \ud \nu_i = 0.
\]
But this implies that $\omega_i$ is trivial in $H(s_0 | \ud)$, contrary to the assumption. Hence $i = \tilde \imath$. In that case, $\omega - \tilde \omega \sim 0$ leads to
\[
 s_0 (\omega_i - \tilde \omega_i) + \ud \nu'_i = 0,
\]
so $\omega_i$ and $\tilde \omega_i$ are indeed representatives of the same cohomology class in $H(s_0 | \ud)$. The proof for $\pi$ proceeds completely analogously.
%\qed
\end{proof}

\end{document}